%% file: DP_MEB_arXiv.tex
\documentclass{article}
\usepackage[final,nonatbib]{my_neurips_2022}

\usepackage[utf8]{inputenc} % allow utf-8 input
\usepackage[T1]{fontenc}    % use 8-bit T1 fonts
\usepackage{url}            % simple URL typesetting
\usepackage{booktabs}       % professional-quality tables
\usepackage{amsfonts}       % blackboard math symbols
\usepackage{nicefrac}       % compact symbols for 1/2, etc.
\usepackage{microtype}      % microtypography
\usepackage{xcolor}         % colors

\usepackage{amsmath,mathtools}
\usepackage{graphicx}
\usepackage{amssymb}
\usepackage{amsthm}
\usepackage{subcaption}
\usepackage{wrapfig}
\usepackage{float}
\usepackage[noend]{algpseudocode}
\usepackage[numbers]{natbib}
\usepackage{xcolor}
\usepackage[colorlinks=true, allcolors=blue]{hyperref}
\hypersetup{
  colorlinks,
  citecolor=black,
  linkcolor=black,
  urlcolor=blue,
  pdftitle={Differentially Private Minimum Enclosing Ball}}
\usepackage{cleveref}
\usepackage{algorithm}

\newcommand{\innerproduct}[2]{\langle #1, #2 \rangle}
\theoremstyle{plain}
\newtheorem{theorem}{Theorem}[section]
\theoremstyle{definition}
 
\newtheorem{claim}[theorem]{Claim}

\newtheorem{lemma}[theorem]{Lemma}
\newtheorem{corollary}[theorem]{Corollary}

\newcommand{\R}{\mathbb{R}}
\newcommand{\N}{\mathbb{N}}
\newcommand{\calM}{\mathcal{M}}
\newcommand{\E}{\mathop{\mathbb{E}}}
\newcommand{\cut}[1]{}

\allowdisplaybreaks

\title{A Differentially Private Linear-Time fPTAS for the Minimum Enclosing Ball Problem}
\author{Bar Mahpud \qquad Or Sheffet\\
Faculty of Engineering\\
  Bar-Ilan University,  Israel \\
  \texttt{\{mahpudb, or.sheffet\}@biu.ac.il} \\}

\begin{document}

\maketitle
\begin{abstract}
The Minimum Enclosing Ball (MEB) problem is one of the most fundamental problems in clustering, with applications in operations research, statistics and computational geometry. In this works, we give the first linear time differentially private (DP) fPTAS for the Minimum Enclosing Ball problem, improving both on the runtime and the utility bound of the best known DP-PTAS for the problem, of Ghazi et al~\cite{GhaziKM20}. Given $n$ points in $\R^d$ that are covered by the ball $B(\theta_{opt},r_{opt})$, our simple iterative DP-algorithm returns a ball $B(\theta,r)$ where $r\leq (1+\gamma)r_{opt}$ and which leaves at most $\tilde O(\frac{\sqrt d}{\gamma\epsilon})$ points uncovered in $\tilde O(\nicefrac n {\gamma^2})$-time. We also give a local-model version of our algorithm, that leaves at most  $\tilde O(\frac{\sqrt {nd}}{\gamma\epsilon})$ points uncovered, improving on the $n^{0.67}$-bound of Nissim and Stemmer~\cite{NissimStemmer18} (at the expense of other parameters). Lastly, we test our algorithm empirically and discuss open problems.
\end{abstract}

\section{Introduction and Related Work}
\label{sec:intro}

One of the fundamental problems in clustering is the Minimum Enclosing Ball (MEB) problem, or the $1$-Center problem, in which we are given a dataset  $P\subset \R^d$ containing $n$ points, and our goal is to find the smallest possible ball $B(\theta_{opt},r_{opt})$ that contains $P$. 
    %In this work, we denote this ball's center as $\theta_{opt}(P)$ and its radius as $r_{opt}(P)$ and omit the dataset whenever it is clear from context.
    %The MEB problem arises in operations research (for example, suppose Planned Parenthood is looking for the best location for their new branch in a certain rural area aiming to set minimum proximity to all clients), in statistics (for example, a CEO has applied certain regression analysis on each employee in her company and now seeks to figure out where the regressors are scattered) and in computational geometry (for example, should one place a ball on a map that covers all home-addresses of people with a certain attribute). 
The MEB problem has applications in various areas of operations research, machine learning, statistics and computational geometry: gap tolerant classifiers~\cite{Burges1998}, tuning Support Vector Machine parameters~\cite{Chapelle2002} and Support Vector Clustering~\cite{BHHSV02, BJS03}, $k$-center clustering~\cite{BHPI02}, solving the approximate $1$-cylinder problem~\cite{BHPI02}, computation of spatial hierarchies (e.g., sphere trees~\cite{Hub96}), and others~\cite{Elzinga1972}. 
The MEB problem is NP-hard to solve exactly, but it can be solved in linear time in constant dimension~\cite{Megiddo83,EppsteinE94} and has several fully-Polynomial Time Approximation Schemes (fPTAS)~\cite{BatdoiuClarkson04, KumarMYY03} that approximate it to any constant $(1+\gamma)$ in time $O(n/\gamma)$; as well as an additive $\gamma$ approximation in sublinear time~\cite{ClarksonHW12}.
    
But in situations where the data is sensitive in nature, such as addresses, locations or descriptive feature-vectors\footnote{Consider a research in a hospital in which one first runs some regression on each patient's data, and then looks for the spread of all regressors of all patients.} we run the risk that approximating the data's MEB might leak information about a single individual. Differential privacy~\cite{DworkMNS06,DworkKMMN06} (DP) alleviates such a concern as it requires that no single individual has a significant effect on the output. Alas, the MEB problem is highly sensitive in nature, since there exist datasets where a change to a single datum may affect the MEB significantly.

In contrast, it is evident that for any fixed ball $B(\theta,r)$ the number of input points that $B$ contains changes by no more than one when changing any single datum. And so, in DP we give \emph{bi-criteria} approximations of the MEB: a ball $B(\theta,r)$ that may leave at most a few points of $P$ uncovered and whose radius is comparable to $r_{opt}$. The work of~\cite{NissimStemmerVadhan16} returns a $O(\sqrt{\log(n)})$-approximation of the MEB while omitting as few as $\tilde O(\nicefrac 1 \epsilon)$ points from $P$, and it was later improved to a  $O(1)$-approximation~\cite{NissimStemmer18}. The work of~\cite{GhaziKM20} does give a PTAS for the MEB problem, but their $(1+\gamma)$-approximation may leave $\tilde O(\nicefrac{\sqrt{d}}{\epsilon\gamma^{3}})$ datapoints uncovered\footnote{See Lemmas 59 \& 60 in~\cite{GhaziKM20}} and it runs in $n^{O(1/\gamma^2)}$-time where the constant hidden in the big-$O$ notation is huge; as it leverages on multiple tools that take $\exp({\sf dim})$-time to construct, such as almost-perfect lattices and list-decodable covers. It should be noted that all of these works actually study the related problem of $1$-cluster in which one is given an additional parameter $t$ and seeks to find the smallest MEB of a subset $Q\subset P$ where $|Q|\geq t$. Lastly (as was first commented in~\cite{GhaziKM20}, Section D.2.1.), a natural way to approximate the MEB problem is through minimizing the convex hinge-loss $L(\theta,x)=\frac 1 r\max\{0, \|x-\theta\|-r\}$ but its utility depends on $r$ (as the utility of DP-ERM scales with the Lipfshitz constant of the loss~\cite{BassilyST14}).\footnote{In fact, there's more to this discussion, as we detail at the end of the introduction.}

By far, one of the most prominent uses of the DP-approximations of the MEB problem lies in range estimation, as $O(1)$-approximations of the MEB can assist in reducing an a-priori large domain to a ball whose radius is proportional to the diameter of $P$. This helps in reducing the $L_2$-sensitivity of problems such as the mean and other distance related queries (e.g. PCA). So for example, if we have $\tilde \Omega(\frac{\sqrt{d}}{\gamma\epsilon})$ points in a ball of radius $10r_{opt}$ then a DP-approximation of the data's mean using the Gaussian mechanism (see Section~\ref{sec:preliminaries}) returns a point of distance $\leq \gamma r_{opt}$ to the true mean (a technique that is often applied in a Subsample-and-Aggregate framework~\cite{nissim2007smooth}). This averaging also gives an efficient $(2+\gamma)$-approximation of the MEB. But it is still unknown whether there exists a DP $c$-approximation of the MEB for $c<2$ whose runtime is below, say, $n^{100}$.

\paragraph{Our Contribution and Organization.} In this work, we give the first DP-fPTAS for the MEB problem. Our algorithm is very simple and so is its analysis. As input, we assume the algorithm is run after the algorithms of~\cite{NissimStemmer18} were already run, and as a ``starting point'' we have both (a) a real number $r_0$ which is a $4$-approximation of $r_{opt}$, and (b) a $10$-approximation of the MEB itself, namely a ball $B$ such that $P\subset B$,\footnote{We can always omit the few input points that may reside outside this ball.} which is centered at a point $\theta_0$ satisfying $\|\theta_0-\theta_{opt}\|\leq 10r_{opt}$.\footnote{We comment that replacing these $4$ and $10$ constants with any other constants merely changes the constants in our analysis in a very straight-forward way.}  It is now our goal to refine these parameters to a $(1+\gamma)$-approximation of the MEB. In fact, we can assume that we have a $(1+\gamma)$-approximation of the value of $r_{opt}$: we simply iterate over all powers: $\frac{r_0}{4}, \frac{r_0}{4}(1+\gamma), \frac{r_0}{4}(1+\gamma)^2, ..., r_0$ where for each guess of $r$ we apply a privacy preserving procedure returning either a point $\theta$ satisfying $P\subset B(\theta, r)$ or $\bot$. In our algorithm we simply use a binary-search over these $O(\nicefrac 1 \gamma)$ possible values, in order to save on the privacy-budget.

Now, given $\theta_0$ and some radius-guess $r$, our goal is to shift $\theta_0$ towards $\theta_{opt}$. So, starting from $\theta^0=\theta_0$, we repeat this simple iterative procedure: we take the mean $\mu$ of the points \emph{uncovered by the current $B(\theta^t,r)$} and update $\theta^{t+1}\gets \theta^t + \frac{\gamma^2}2(\mu-\theta^t)$. We argue that, if $r\geq r_{opt}$ then after $T=O(\gamma^{-2}\log(\nicefrac 1 \gamma))$-iterations we get $\theta^T$ such that $\|\theta^T-\theta_{opt}\|\leq \gamma r_{opt}$ and therefore have that $P\subset B(\theta^T, (1+\gamma)r)$. The reason can be easily seen from Figure~\ref{fig:uncovered_point_projection}~--- any point $x\in P$ which is uncovered by the current $B(\theta^t,r)$ must be closer to $\theta_{opt}$ than to $\theta^t$, and therefore must have a noticeable projection onto the direction $\theta_{opt}-\theta^t$. Thus, in a Perceptron-like style, making a $\Theta(\gamma^2)$-size step towards this $x$ must push us significantly in the $\theta_{opt}-\theta^t$ direction. We thus prove that if the distance of $\theta^t$ from $\theta_{opt}$ is large, this update step reduces our distance to $\theta_{opt}$. Note that our proof shows that in the non-private case it suffices to take any uncovered point in order to make this progress, or any convex-combination of the uncovered points. 

In the private case, rather than using the true mean of the uncovered points in each iteration, we have to use an approximated mean. So we prove that applying our iterative algorithm with a ``nice'' distribution whose mean has a large projection in the $\theta_{opt}-\theta^t$ direction also returns a good $\theta^T$ in expectation, and then amplify the success probability by na\"ive repetitions. We also give a SQ-style algorithm for approximating the MEB under proximity conditions between the true- and the noisy-mean, a result which may be of interest by itself. 
After discussing preliminaries in Section~\ref{sec:preliminaries}, 
we present both the standard (non-noisy) version of our algorithm and its noisy variation in Section~\ref{sec:MEB_fPTAS}.

Having established that our algorithm works even with a ``nice'' distribution whose mean approximates the mean of the uncovered points, all that is left is just to set the parameters of a privacy preserving algorithm accordingly. To that end we work with the notion of zCDP~\cite{BunS16} and apply solely the Gaussian mechanism. To obtain these nice properties, it follows that the  number of uncovered points must be $\Omega(\nicefrac{\sqrt{d}}{\epsilon^t})$ where $\epsilon^t$ is the privacy budget of the $t^{\rm th}$-iteration, or else we halt. And due to the composition theorem of DP it suffices to set $\epsilon^t = O(\nicefrac{\epsilon}{\sqrt T})$. This leads to a win-win situation: either we find in some iteration a ball that leaves no more than $\tilde O(\nicefrac{\sqrt{d}}{\gamma\epsilon})$ points uncovered, or we complete all iterations and obtain a ball of radius $\leq (1+\gamma)r$ that covers all of $P$. The full details of this analysis appear in Section~\ref{sec:DP_MEB_fPTAS}.
%but due to space constraints we deferred its full proofs to the Supplementary Material, Section~\ref{apx_sec:proofsDPAlg}. 
We then repeat this analysis but in the local-model, where each user adds Gaussian noise to her own  input point. This leads to a similar analysis incurring a $\sqrt{n}$-larger bounds, as detailed in Section~\ref{sec:local_DP_MEB_fPTAS}.
%, but its analysis is deferred to the Supplementary Material, Section~\ref{apx_sec:proofsLDPAlg}.

While at the topic of local-model DP (LDP) algorithms, it is worth mentioning that the algorithms of~\cite{NissimStemmer18}, which provide us with a good initial ``starting point'', do have a LDP-variant. Yet the LDP variants of these algorithms may leave as many as $n^{0.67}$ datapoints uncovered. So in  Appendix~\ref{apx_sec:preliminary_algorithms} we give simple differentially private algorithms (in both the curator- and local-models) that obtain such good $\theta_0$ and $r_0$. Formally, our LDP-algorithm returns a ball $B(\theta_0,r_0)$ s.t. by projecting all points in $P$ onto $B(\theta_0,r)$ we alter no more than $\tilde O(\sqrt{d}/\epsilon)$ points and obtain $P'\subset B(\theta_0, r_0)$ where $r_0\leq 6r_{opt}(P')$. Thus, combining our LDP algorithm for finding a good starting point together with the algorithm of Section~\ref{sec:local_DP_MEB_fPTAS} we get an overall $(1+\gamma)$-approximation of the MEB in the local model which may omit / alter as many as $\tilde O(\nicefrac{\sqrt{nd}}{\gamma\epsilon})$-points. We comment that while this improves on the previously best-known LDP algorithm's bound of $n^{0.67}$, our algorithm's dependency on parameters such as the dimension $d$ or grid-size\footnote{It is known~\cite{BunNSV15} that DP MEB-approximation requires the input points to lie on some prespecified finite grid.} is worse, and furthermore~-- that the analysis of~\cite{NissimStemmer18} (i) relates to the problem of $1$-cluster (finding a cluster containing $t\leq n$ many points) and (ii) separates between the required cluster size and the number of omitted points (which is much smaller and only logarithmic in $d$), two aspects that are not covered in our work.

Lastly, we provide empirical evaluations of our algorithm in Section~\ref{sec:experiments}
%(which we deferred in its entirety to the Supplementary Material, Section~\ref{sec:experiments}) 
showing a rather ubiquitous performance across multiple datasets, and discuss open problems in Section~\ref{sec:conclusion}. 

\paragraph{Comparison with the ERM Baseline.} Recall that the MEB problem, given a suggested radius $r$ and a convex set $\Theta$, can be formulated as a ERM problem using a hinge-loss function $\ell^1(\theta;x) = \max\{ 0, \frac{\|x-\theta\|-r}{{\rm diam}(\Theta)} \}$. Indeed, when ${{\rm diam}(\Theta)}\gg r$ then privately solving this ERM problem gives no useful guarantee about the result, but much like our algorithm one can first find some $\theta_0$ close up to, say, $10r$ to $\theta_{opt}$ and set $\Theta$ as a ball of radius $O(r)$. Since there exists $\theta_{opt}$ for which $\frac 1 n\sum_{x}\ell^1(\theta; x) = 0$, then private SGD~\cite{BassilyST14,BassilyFTT19} returns $\tilde \theta$ for which 
$\frac 1 n\sum_{x}\ell^1(\tilde \theta; x)\leq C\frac{\sqrt d}{\epsilon n}$ for some constant $C>0$. This upper-bounds the number of points that contribute $\gamma r$ to this loss at $C\frac{\sqrt d}{\epsilon \gamma}$, and so $|P\setminus B(\tilde\theta,(1+\gamma)r)| = O(\frac{\sqrt d}{\epsilon \gamma})$. However, the caveat is that the SGD algorithm achieves such low loss using $O(n^2)$-SGD iterations.\footnote{Unfortunately, the hinge-loss isn't smooth, ruling out the linear SGD of~\cite{FeldmanKT20}.} In contrast our analysis can be viewed as proving that for the equivalent ERM in the square of the norm, $\ell^2(\theta;x) = \max\{ 0, \frac{\|x-\theta\|^2-r^2}{{\rm diam}(\Theta)^2} \}$, it suffices to make only $\tilde O(\gamma^{-2})$ \emph{non-zero} gradient steps to have some $\theta^T$ s.t. $\|\theta^T-\theta_{opt}\|\leq \gamma r$ so that $B(\theta^T,(1+\gamma)r)$ covers all of the input. Thus, our result is obtained in linear $\tilde O(n/\gamma^2)$-time.

\section{Preliminaries}
\label{sec:preliminaries}

\paragraph{Notation.} Given a vector $v\in \R^d$ we denote its $L_2$-norm as $\|v\|$, and also use $\langle v, u \rangle$ to denote the dot-product between two $d$-dimensional vectors $u$ and $v$. A (closed) ball $B(\theta,r)$ is the set of all points $B(\theta, r)=\{x\in \R^d: \|x-\theta\|\leq r\}$. We use $\tilde O(\cdot)$ / $\tilde \Omega(\cdot)$ to denote big-$O$ / big-$\Omega$ dependency up to ${\rm poly}\log$ factors. We comment that in our work we made no effort to optimize constants.

\paragraph{The Gaussian and $\chi^2_d$-Distributions.} 
Given two parameters $\mu\in \R$ and $\sigma^2>0$ we denote ${\cal N}(\mu,\sigma^2)$ as the Gaussian distribution whose {\sf PDF} at a point $x\in \R$ is $(2\pi\sigma^2)^{0.5}\exp(-\frac{(x-\mu)^2}{2\sigma^2})$. Standard concentration bounds give that for any $x>1$ the probability $\Pr_{X\sim {\cal N}(\mu,\sigma^2)}[|X-\mu|\geq x\sigma]\leq 2\exp(-x^2/2)$. It is well-known that given two independent random variable $X\sim {\cal N}(\mu_1, \sigma_1^2)$ and $Y\sim {\cal N}(\mu_2,\sigma_2^2)$ their sum is distributed like a Gaussian $X+Y\sim {\cal N}(\mu_1+\mu_2, \sigma_1^2+\sigma_2^2)$. We also denote ${\cal N}(v,\sigma^2I_d)$ as the distribution over $d$-dimensional vectors where each coordinate $j$ is drawn i.i.d.~from ${\cal N}(v_j,\sigma^2)$. Given $X\sim {\cal N}(0,\sigma^2I_d)$ it is known that $\|X\|^2$ is distributed like a $\chi^2_d$-distribution; and known concentration bounds on the $\chi^2_d$-distribution give that for any $x>1$ the probability $\Pr_{X\sim{\cal N}(0,\sigma^2I_d)}[\|X\|^2 > \sigma^2(\sqrt{d}+x)^2]\leq \exp(-x^2/2)$.

\paragraph{Differential Privacy.}
Given a domain ${\cal X}$, two multi-sets $P,P'\in {\cal X}^n$ are called \emph{neighbors} if they differ on a single entry. An algorithm (alternatively, mechanism) $\calM$ is said to be \emph{$(\epsilon,\delta)$-differentially private} (DP)~\cite{DworkMNS06,DworkKMMN06} if for any two neighboring $P,P'$ and any set $S$ of possible outputs we have: $\Pr[\calM(P)\in S]\leq e^\epsilon \Pr[\calM(P')\in S] + \delta$.

An algorithm is said to be $\rho$-zero concentrated differentially privacy (zCDP)~\cite{BunS16} if for and two neighboring $P$ and $P'$ and any $\alpha>1$, the $\alpha$-R\'eyni divergence between the output distribution of $\calM(P)$ and of $\calM(P')$ is upper bounded by $\alpha\rho$, namely
\[ \forall \alpha>1, ~~ \frac{1}{\alpha-1}\log\left(\E_{x\sim \calM(P')}\left[\left(\frac{{\sf PDF}[\calM(P)=x]}  {{\sf PDF}[\calM(P')=x]}\right)^\alpha\right]\right) \leq \alpha\rho\]
It is a well-known fact that the composition of two $\rho$-zCDP mechanisms is $2\rho$-zCDP. It is also known that given a function $f:{\cal X}^n\to\R^d$ whose $L_2$-global sensitivity is $\max_{P\sim P'}\|f(P)-f(P')\|_2 \leq G$ then the Gaussian mechanism that returns $f(D)+X$ where $X\sim {\cal N}(0, \frac{G^2}{2\rho}I_d)$ is $\rho$-zCDP. Lastly, it is known that any $\rho$-zCDP mechanism is $(\epsilon,\delta)$-DP for any $\delta<1$ and $\epsilon = \rho + \sqrt{4\rho\ln(1/\delta)}$. This suggests that given $\epsilon\leq 1$ and $\delta \leq e^{-2}$ it suffices to use a $\rho$-zCDP mechanis with $\rho \leq \frac{\epsilon^2}{5\ln(1/\delta)}$.

The \emph{Local-Model} of DP: while standard algorithms in DP assume the existence of a trusted curator who has access to the raw data, in the local-model of DP no such curator exists. While the formal definition of the local-model involves the notion of protocols (see~\cite{Vadhan17} for a formal definition), for the context of this work it suffices to say each respondent randomized her own messages so that altogether they preserve $\rho$-zCDP.

\section{A Non-Private fPTAS for the MEB Problem}
\label{sec:MEB_fPTAS}

In this section we give our non-private algorithm. We first analyze it assuming no noise~-- namely, in each iteration we use the precise mean of the points that do not reside inside the ball $B(\theta^t,r)$. Later, in Section~\ref{subsec:noisy_version} we discuss a version of this algorithm in which rather than getting the exact mean, we get a point which is sufficiently close to the mean.

\begin{algorithm}
    \caption{Non-Private Minimum Enclosing Ball\label{alg:np-meb}}
     \hspace*{\algorithmicindent}\textbf{Input:} a set of $n$ points $P \subseteq \mathbb{R}^d$, an approximation parameter $\gamma \in (0,1)$, \\ \hspace*{\algorithmicindent} an initial radius $r_0$ s.t. $r_{opt}\leq r_0 \leq 4r_{opt}$, and an initial center $\theta_0$ s.t. $\|\theta_0 - \theta_{opt}\| \leq 10r_{opt}$.
    \begin{algorithmic}[1]
     % Input:
     % Output:
     \State Set $i_{\min} \gets 0$, $i_{\max} \gets \ln_{1+\gamma}(4)(\approx \frac{4}{\gamma})$, and $\theta^*\gets\theta_0$.
     \While{($i_{\min}<i_{\max}$)}
        \State $i_{cur} = \lfloor\frac{i_{\min}+i_{\max}}2\rfloor$
        \State $r_{cur} \leftarrow  (1+\gamma)^{i_{cur}}\cdot r_0 / 4$
        \State $\theta_{cur} \gets {\hyperref[alg:gd-meb]{\textrm{MMEB}}}(P, \gamma, r_{cur}, \theta_0)$
        \If {$P\subset B(\theta_{cut}, (1+\gamma)r_{cur})$} 
             \State Set $i_{\max}\gets i_{cur}$, $\theta^*\gets \theta_{cur}$ and $r^*\gets (1+\gamma)r_{cur}$
             \Else
             \State $i_{\min} \gets i_{cur}+1$
        \EndIf
     \EndWhile
     \State \Return $B(\theta^*, r^*)$
    \end{algorithmic}
\end{algorithm}

\begin{algorithm}
    \caption{Margin based Minimum Enclosing Ball (MMEB)\label{alg:gd-meb}}
     \hspace*{\algorithmicindent}\textbf{Input:} a set of $n$ points $P \subseteq \mathbb{R}^d$, an approximation parameter $\gamma \in (0,1)$, \\ \hspace*{\algorithmicindent} a candidate radius $r$, and an initial center $\theta_0$ s.t. $\|\theta_0 - \theta_{opt}\| \leq 10r_{opt}$.
    \begin{algorithmic}[1]
     % Input:
     % Output:
     \State Set $T \leftarrow \frac{4}{\gamma^2}\ln(\frac{100}{\gamma^2})$, and $\theta^0=\theta_0$.
     \For{$t=0, 1, 2, \ldots, T-1$}
        \If{($\{x \in P: x \notin B(\theta^t, r) \} = \emptyset$)}
            \Return $\theta^t$         \Else 
        \State Set $n^t_{w} \leftarrow |\{x \in P: x \notin B(\theta^t, r) \}|$
        and $\mu^t_{w} \leftarrow \frac{1}{n^t_{w}} \sum\limits_{x \notin B(\theta^t, r)} x$
        \State Update $\theta^{t+1} \leftarrow \theta^t - \frac{\gamma^2}{2}(\theta^t - \mu^t_{w})$
        \EndIf
     \EndFor
     \State \Return $\theta^T$
    \end{algorithmic}
\end{algorithm}

\begin{theorem}
    \label{thm:Alg_no_noise_fptas}
    For any $P\subset \R^d$, denote $B(\theta_{opt}, r_{opt})$ as the MEB of $P$. Then Algorithm~\ref{alg:np-meb} returns a ball $B(\theta,r)$ where $P\subset B(\theta,r)$ and $r\leq (1+3\gamma)r_{opt}$.
\end{theorem}

At the core of the proof of Theorem~\ref{thm:Alg_no_noise_fptas} lies the following lemma.

\begin{lemma}
    \label{lem:MMEB_returns_close_point}
    Applying Algorithm~\ref{alg:gd-meb} with any $r\geq r_{opt}$ and any $\theta_0$ where $\|\theta_0-\theta_{opt}\|\leq 10r_{opt}$ we obtain a $\theta$ where $\|\theta-\theta_{opt}\|\leq \gamma r_{opt}$ in at most $T$ iterations. 
\end{lemma}

It is important to note that Lemma~\ref{lem:MMEB_returns_close_point} holds even if in each iteration the update step isn't based on the mean $\mu^t$ of the set of uncovered point, but rather \emph{any} convex combination of the uncovered points. Specifically, even if we use in each iteration a single point which is uncovered by $B(\theta^t,r)$, then the algorithm's convergence in $T$ steps can be guaranteed.

\begin{proof}[Proof of Theorem~\ref{thm:Alg_no_noise_fptas}]
    Suppose Lemma~\ref{lem:MMEB_returns_close_point} indeed holds. Then it immediately implies whenever Algorithm~\ref{alg:gd-meb} is run with $r\geq r_{opt}$ we obtain a point $\theta$ where $P\subset B(\theta_{opt},r_{opt}) \subset B(\theta, (1+\gamma)r_{opt})$. Denote $i^* = \min\{i\in \N: \frac {r_0}4 (1+\gamma)^i \geq r_{opt}\}$. It is simple to prove inductively that in each iteration of Algorithm~\ref{alg:np-meb} we have that $i^* \geq i_{\min}$. Next, call an integer $i$ successful if we obtain for its radius $r_{cur}(i)$ some point $\theta$ where $P\subset B(\theta, (1+\gamma)r_{cur}(i))$. Again, it is simple to argue inductively that $i_{\max}$ is always successful.
    It follows that when the binary search of Algorithm~\ref{alg:np-meb} terminates, $i_{\min}=i_{\max}$ and we have a successful $i$, and so we return a ball of radius $\frac{r_0}4 (1+\gamma)^{i_{\min}}\cdot (1+\gamma) \leq (1+\gamma)^2 r_{opt} \leq (1+3\gamma)r_{opt}$ which contains all points in $P$, thus concluding our proof.
\end{proof}

\begin{wrapfigure}{r}{0.52\textwidth}
    \centering
    \vspace{-.65cm}
    \includegraphics[scale=0.58]{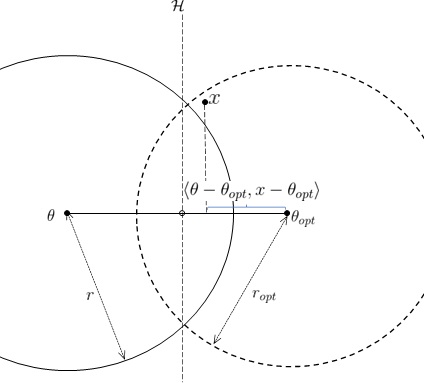}
    \caption{\label{fig:uncovered_point_projection} For a point $x$ uncovered by $B(\theta,r)$ where $r\geq r_{opt}$, it must be that $x$'s projection onto the $\overline{\theta\theta_{opt}}$-line is closer to $\theta_{opt}$ than to $\theta$.}
    \vspace{-1cm}
\end{wrapfigure}

Thus, all that is left is to prove Lemma~\ref{lem:MMEB_returns_close_point}. Its proof, in turn, requires the following claim.

\begin{claim}\label{clm:pt_outside_ball_dotproduct}
    Given a set of $n$ points $P \subseteq \mathbb{R}^d$, let $B(\theta_{opt}, r_{opt})$ denote the MEB of $P$. Let $\theta \in \mathbb{R}^d$ be an arbitrary point, and let $r$ be any real number where $r\geq r_{opt}$. Then for any $x \in P$ s.t.~$\|\theta - x\| > r$ it holds that \[\innerproduct{\theta - \theta_{opt}}{x - \theta_{opt}} \leq \frac{1}{2}\|\theta - \theta_{opt}\|^2\]
\end{claim}

\begin{proof}
    Let $x\in P$ be a point s.t.~$x\notin B(\theta,r)$, as depicted in Figure~\ref{fig:uncovered_point_projection}.
    %Consider the plane determined by $\theta, \theta_{opt}$ and $x$.
    Let $m$ be the middle point $\frac {\theta + \theta_{opt}}{2}$, and let $\mathcal{H}$ be the hyperplane orthogonal to $\theta - \theta_{opt}$ which passes through $m$. Denote $\mathcal{H}^+$ as the (open) half-space $\mathcal{H}^+ = \{ z \in \mathbb{R}^d : \|z - \theta_{opt}\| < \|z - \theta\| \}$. Therefore $x \in \mathcal{H}^+$ which in turn implies that 
    \[\innerproduct{x - \theta_{opt}}{\theta-\theta_{opt}} < \innerproduct{m - \theta_{opt}}{\theta-\theta_{opt}} = \frac{1}{2}\|\theta - \theta_{opt}\|^2 \ \hfill\qedhere\]
\end{proof}

%\begin{claim}
%\label{clm:progress_of_alg_2}
%    Consider any iteration $t$ of Algorithm~\ref{alg:gd-meb} where $\{x \in P: x \notin B(\theta^t, r) \} \neq \emptyset$. Then it holds that $\|\theta^{t+1} - \theta_{opt}\|^2 \leq (1 - \frac{\gamma^2}{2})\|\theta^{t} - \theta_{opt}\|^2 + (\frac{\gamma^2}{2})^2\cdot r_{opt}^2$
%\end{claim}
We are now ready to prove our main lemma.

\begin{proof}[Proof of Lemma~\ref{lem:MMEB_returns_close_point}]
First, we argue that in any iteration $t$ of Algorithm~\ref{alg:gd-meb} where $\{x \in P: x \notin B(\theta^t, r) \} \neq \emptyset$ it holds that $\|\theta^{t+1} - \theta_{opt}\|^2 \leq (1 - \frac{\gamma^2}{2})\|\theta^{t} - \theta_{opt}\|^2 + (\frac{\gamma^2}{2})^2\cdot r_{opt}^2$. That is because by definition
    \begin{align}
    \|\theta^{t+1} - \theta_{opt}\|^2 &= \left\| \left((1 - \frac{\gamma^2}{2})\theta^t + \frac{\gamma^2}{2} \mu^t_{w}\right)- \theta_{opt} \right\|^2 
    = \left\| (1 - \frac{\gamma^2}{2})\left(\theta^t-\theta_{opt}\right) + \frac{\gamma^2}{2} \left(\mu^t_{w}- \theta_{opt}\right) \right\|^2
    \cr &= (1 - \frac{\gamma^2}{2})^2 \cdot \|\theta^t - \theta_{opt}\|^2 
     + 2\frac{\gamma^2}{2}(1 - \frac{\gamma^2}{2}) \innerproduct{\theta^t - \theta_{opt}}{\mu_{w^t} - \theta_{opt}}  + (\frac{\gamma^2}{2})^2\cdot \|\mu_{w^t} - \theta_{opt}\|^2 \notag
     \intertext{
     Claim~\ref{clm:pt_outside_ball_dotproduct} gives that $\innerproduct{\theta^t - \theta_{opt}}{\mu^t_{w} - \theta_{opt}} = \frac 1 {n^t_w}\sum\limits_{x\notin B(\theta^t,r)}\innerproduct{\theta^t - \theta_{opt}}{x - \theta_{opt}} \leq \frac{1}{2}\|\theta^t - \theta_{opt}\|^2$, so
     }
    & \leq  (1 - \frac{\gamma^2}{2})^2 \cdot \|\theta^t - \theta_{opt}\|^2 
           + 2(\frac{\gamma^2}{2} - \frac{\gamma^4}{4})\cdot \frac{1}{2}\|\theta^t - \theta_{opt}\|^2  + (\frac{\gamma^2}{2})^2\cdot \|\mu_{w^t} - \theta_{opt}\|^2 \notag
    \intertext{Lastly note that the ball $B(\theta_{opt},r_{opt})$ is convex and so}
    &\leq (1 - \gamma^2 + \frac{\gamma^4}{4}) \cdot \|\theta^t - \theta_{opt}\|^2 
           + (\frac{\gamma^2}{2} - \frac{\gamma^4}{4})\cdot \|\theta^t - \theta_{opt}\|^2  + \frac{\gamma^4}{4}\cdot r_{opt}^2
           \cr & \leq (1 - \frac{\gamma^2}{2})\|\theta^{t} - \theta_{opt}\|^2 + \frac{\gamma^4}{4}\cdot r_{opt}^2 \label{eq:progress_of_algorithm_MMEB}
    \end{align}
    So now, consider any iteration of Algorithm~\ref{alg:gd-meb} with $r\geq r_{opt}$ and where $\|\theta^t-\theta_{opt}\|\geq \gamma r_{opt}$ and in which we make an update step. Due to Equation~\eqref{eq:progress_of_algorithm_MMEB}
\begin{align*}
    \|\theta^{t+1} - \theta_{opt}\|^2 &\leq (1 - \frac{\gamma^2}{2})\|\theta^{t} - \theta_{opt}\|^2  + \frac{\gamma^4}{4}\cdot r_{opt}^2  
     \leq (1 - \frac{\gamma^2}{2})\|\theta^{t} - \theta_{opt}\|^2  + \frac{\gamma^4}{4}\cdot \frac{\|\theta^t - \theta_{opt}\|^2}{\gamma^2} 
    \cr & = (1 - \frac{\gamma^2}{4})\|\theta^{t} - \theta_{opt}\|^2 \leq e^{- \frac{\gamma^2}{4}}\|\theta^{t} - \theta_{opt}\|^2
\end{align*}
This suggests that after $T = \frac{4}{\gamma^2}\ln(\frac{100}{\gamma^2})$ iterations where $\|\theta^t-\theta_{opt}\|\geq \gamma r_{opt}$ we get that
\[  \|\theta^T-\theta_{opt}\|^2 \leq e^{- \frac{T\gamma^2}{4}}\|\theta_{0} - \theta_{opt}\|^2 \leq \frac{\gamma^2}{100} \cdot 100 r^2_{opt}  = \gamma^2 r^2_{opt}   \]
as required. Now, should it be the case that in some iteration $\|\theta^t-\theta_{opt}\| < \gamma r_{opt}$ and we make an update step. Again, Equation~\eqref{eq:progress_of_algorithm_MMEB} asserts that 
\begin{align*}
    \|\theta^{t+1} - \theta_{opt}\|^2 &\leq (1 - \frac{\gamma^2}{2})\|\theta^{t} - \theta_{opt}\|^2  + \frac{\gamma^4}{4}\cdot r_{opt}^2 
    < (1 - \frac{\gamma^2}{2})\gamma^2 r_{opt}^2  + \frac{\gamma^4}{4}\cdot r_{opt}^2
    < \gamma^2 r_{opt}^2
\end{align*}
so once $\|\theta^t-\theta_{opt}\|<\gamma r_{opt}$ then we have that $\|\theta^\tau-\theta_{opt}\|<\gamma r_{opt}$ for all $\tau \geq t$.
\end{proof}

We comment that non-privately, it is rather simple to obtain a good $r_0$ and a good starting point $\theta_0$: $r_0 = {\rm diam}(P)$ which is known to be upper bounded by $2r_{opt}$ and $\theta_0$ can be any $x\in P$ which is within distance $r_{opt}$ from the true center of the MEB of $P$. Next, we comment that Algorithm~\ref{alg:gd-meb} runs in time $O(T\cdot n)$ since the averaging of the points in $P\setminus B(\theta^t,r)$ takes $O(n)$-time na\"ively. Thus, overall, the runtime of Algorithm~\ref{alg:np-meb} is $O(nT\log(\nicefrac 1 \gamma))=O(n \frac{\log^2(1/\gamma)}{\gamma^2})$. Lastly, we comment that in Algorithm~\ref{alg:gd-meb} we could replace the mean $\mu_w^t$ of the uncovered points with any convex combination (even a single $x\notin B(\theta^t,r)$) and the analysis carries through. This implies that the ERM discussed in the introduction (with based on the loss-function $\ell^2$)  requires a constant step-rate and can halt after $\tilde O(\gamma^{-2})$ iterations of \emph{non-zero} gradients.

\subsection{The Noisy/SQ-Version of the fPTAS for the MEB Problem}
\label{subsec:noisy_version}

Now, we consider  a scenario where in each iteration $t$, rather than using the exact mean $\mu_w^t = \frac{\sum_{x\in P\setminus B(\theta^t,r)}x}{|P\setminus B(\theta^t,r)|}$, we obtain an approximated mean
$\tilde \mu_w^t = \mu_w^t + \Delta^t$. We consider here two scenarios: (a) where $\Delta^t$ is a zero-mean bounded-variance random noise~--- a setting we refer to as \emph{random noise} from now own; and (b) where $\Delta^t$ is an arbitrary noise subject to the constraint that $\|\Delta^t\| = O(\gamma r)$~--- a setting we refer to as \emph{arbitrary small noise}. Since the latter isn't used in our algorithm we defer it to Appendix~\ref{apx_sec:SQNoise}.

\paragraph{The random noise setting.} In this setting, our update step in each iteration is made not using a deterministically chosen uncovered point but rather by a draw from a distribution ${\cal D}^t$ whose mean is ``as good'' as an uncovered point. This requires us to make two changes to the algorithm: (i) modify the update rate and (ii) repeat the entire algorithm $R = O(\log(1/\beta))$ times. 

\newcommand{\calD}{\mathcal{D}}
\begin{claim}
\label{clm:random_step_algorithm}
Consider an altered version of Algorithm~\ref{alg:gd-meb} which (1) repeats the algorithm $R = \lceil\log_{\nicefrac 4 3}(1/\beta)\rceil$ times, (2) each repetition is composed of at most $T = \frac {4096}{\gamma^2} \ln(\frac{121\cdot 4}{\gamma^2})$ update-steps and (3) in each iteration where it doesn't terminate it draws a point $z\sim \calD^t$ and makes that update-step: $\theta^{t+1}\gets \theta^t + \frac{\gamma^2}{2048}z$. If it holds that for each iteration $t$ we have that $\calD^t$ satisfies the  two properties
\begin{align}
    \textrm{(i)} &\E\limits_{z \sim \mathcal{D}^t}\left[\innerproduct{\theta_{opt} - \theta^t} {z}~|~ \theta^t \right] \geq \frac{1}{4}\|\theta^t - \theta_{opt}\|^2\cr 
    \textrm{(ii)}  &\E\limits_{z \sim \mathcal{D}^t}[\|z\|^2~|~ \theta^t] \leq 512r^2 \label{eq:requirements_of_good_dist}
\end{align}
then, provided that $r\geq r_{opt}$, we have that w.p. $\geq 1-\beta$ one of the $R$ repetitions of the revised algorithm returns a candidate center $\theta^T$ where  $P\subset B(\theta^T, (1+\gamma)r)$. 
\end{claim}

\begin{proof}
    To prove the claim it suffices to show that in a single execution of the algorithm we have that $\Pr[\|\theta^T-\theta_{opt}\|\leq \gamma r] = \Pr[\|\theta^T-\theta_{opt}\|^2\leq \gamma^2 r^2]\geq \nicefrac 1 4$, implying that in $R$ repetitions of the algorithm the failure probability decreases to $(\nicefrac 3 4)^R=\beta$. To that end, denote the non-negative random variables $Y^t = \|\theta^t-\theta_{opt}\|^2$ for each iteration $t$. Note that if we show that $\E[Y^T]\leq \frac 3 4 \gamma^2 r^2$ then Markov's inequality implies that $\Pr[Y^T\geq \gamma^2 r^2]\leq \nicefrac 3 4$. So our goal is to prove that $\E[Y^T]\leq \frac 3 4 \gamma^2 r^2$. 

    We can now analyze the conditional expectation and observe that 
    \begin{align*}
    &\E\left[\|\theta^{t+1} - \theta_{opt}\|^2 ~|~ \theta^t\right] = \E\left[\left\|\theta^t - \theta_{opt} + \frac{\gamma^2}{2048}z\right\|^2~|~ \theta^t\right]
    \cr & ~~~ = \E\left[\|\theta^t - \theta_{opt}\|^2 + \frac{2\gamma^2}{2048}\innerproduct{z}{\theta^t - \theta_{opt}} + (\frac{\gamma^2}{2048})^2\|z\|^2~|~ \theta^t\right]
    \cr & ~~~ \stackrel{z \sim \mathcal{D}^t}\leq \|\theta^t - \theta_{opt}\|^2 - \frac{2\gamma^2}{2048}\cdot \frac{1}{4}\|\theta^t - \theta_{opt}\|^2 + \frac{\gamma^4\cdot 512r^2}{2048^2}
     = (1 - \frac{\gamma^2}{4096})\|\theta^t - \theta_{opt}\|^2 + \frac{\gamma^4}{8192}r^2
    \end{align*}
    Since $\E[Y^{t+1} ~|~ \theta^t] \leq (1 - \frac{\gamma^2}{4096})Y^t + \frac{\gamma^4}{8192}r^2$ then
it is easy to see that $\E[Y^T]\leq (1-\frac{\gamma^2}{4096})^T \cdot Y^0 + \frac{\gamma^4}{8192}r^2 \sum_{t=0}^{T-1}(1-\frac{\gamma^2 }{4096})^t \leq (1-\frac{\gamma^2}{4096})^t \cdot (11r)^2 + \frac{\gamma^2}{2}r^2$. It follows that iteration $T=\frac {4096}{\gamma^2} \ln(\frac{121\cdot 4}{\gamma^2})$ we have that $\E[Y^T]\leq \frac{\gamma^2}{4}r^2 + \frac{\gamma^2}2 r^2 = \frac{3}{4}\gamma^2 r^2$ as required.
\end{proof}

\begin{corollary}
    \label{cor:random_step_algorithm_whp}
    Suppose that in each iteration $t$ of the revised algorithm $\calD^t$ is a distribution that satisfies the required two properties of Claim~\ref{clm:random_step_algorithm} w.p. $\geq 1 - \frac{1}{8T \cdot \lceil\log_{8/7}(1/\beta)\rceil }$. Then, repeating this algorithm $R = \lceil\log_{8/7}(1/\beta)\rceil$ many times we have that w.p. $\geq 1-\beta$ it holds that for at least one repetition we have $P\subset (B(\theta^T,(1+\gamma)r)$.
\end{corollary}
\begin{proof}
    Using the union bound, it follows that in one of the $R\cdot T$ repetition of the revised algorithm the probability that one draw isn't from a  good $\calD^t$ (that does satisfy these two properties) is at most $\nicefrac 1 8$. It follows that $\Pr[Y^T\geq \gamma^2 r^2]\geq \nicefrac 1 4 - \nicefrac 1 8=\nicefrac 1 8$. Repeating this algorithm $R$ reduces the failure probability to $(\nicefrac 7 8)^R\leq \beta$.
\end{proof}

%\paragraph{The arbitrary small noise setting.} 
\DeclareRobustCommand{\SQNoiseParagraph}{For completeness, we also bring the SQ-model version of the algorithm where in each iteration we obtain an approximated center $\tilde \mu^t$ where $\Delta^t =  \tilde \mu_w ^t- \mu^t_w$ is of magnitude propostional to $\gamma r$. 
We modify Algorithm~\ref{alg:gd-meb} so that our update scale shrinks by a constant factor to $\gamma^2/8$, namely we set $\theta^{t+1}\gets (1-\frac{\gamma^2}8)\theta^t + \frac{\gamma^2}{8}\tilde \mu_w^t$. We now prove that the revised algorithm still converges to a point close to $\theta_{opt}$.

\begin{lemma}
    \label{lem:noisy_MMEB_returns_close_point}
    Applying Algorithm~\ref{alg:gd-meb} with any $4r_{opt}\geq r\geq r_{opt}$ and any $\theta_0$ where $\|\theta_0-\theta_{opt}\|\leq 10r_{opt}$, where in each iteration we use an approximated mean $\tilde \mu_w^t=\mu_w^t + \Delta^t$ where $\|\Delta^t\|\leq \frac{\gamma r}{16}  \leq \frac{\gamma r_{opt}} 4$  we obtain a $\theta$ where $\|\theta-\theta_{opt}\|\leq \gamma r_{opt}$ in at most $16T=\frac{64}{\gamma^2}\ln(\nicefrac{100}{\gamma^2})$ iterations. 
\end{lemma}
\begin{proof}
    First, analogously to Lemma~\ref{lem:MMEB_returns_close_point} we have that in each update step we get
        \begin{align*}
    \|\theta^{t+1} - \theta_{opt}\|^2 &= \left\| \left((1 - \frac{\gamma^2}{8})\theta^t + \frac{\gamma^2}{8} \tilde\mu^t_{w}\right)- \theta_{opt} \right\|^2 
    = (1 - \frac{\gamma^2}{8})^2 \cdot \|\theta^t - \theta_{opt}\|^2 
     \cr&~~ + 2\frac{\gamma^2}{8}(1 - \frac{\gamma^2}{8}) \left(\innerproduct{\theta^t - \theta_{opt}}{\mu_{w}^t - \theta_{opt}}+\innerproduct{\theta^t-\theta_{opt}}{\Delta^t}\right)  + (\frac{\gamma^2}{8})^2\cdot \|\mu_{w}^t - \theta_{opt}+\Delta^t\|^2
    \cr & \leq  (1 - \frac{\gamma^2}{8})^2 \cdot \|\theta^t - \theta_{opt}\|^2 
           + 2(\frac{\gamma^2}{8} - \frac{\gamma^4}{64})\cdot \left(\frac{1}{2}\|\theta^t - \theta_{opt}\|^2 + \|\theta^t - \theta_{opt}\|\cdot \frac{\gamma r_{opt}}{4}\right) 
           \cr & ~~ + (\frac{\gamma^2}{8})^2\cdot \left( 2\|\mu_{w}^t - \theta_{opt}\|^2+2\frac{\gamma^2r_{opt}^2}{4^2}\right)
           \cr & \leq (1 - \frac{\gamma^2}{8})^2\|\theta^{t} - \theta_{opt}\|^2 
           + 2(\frac{\gamma^2}{8} - \frac{\gamma^4}{64})\cdot \|\theta^t - \theta_{opt}\|\left(\frac{1}{2}\|\theta^t - \theta_{opt}\| + \frac{\gamma r_{opt}}{4}\right) + \frac{3\gamma^4}{64}r_{opt}^2
    \end{align*}
    It follows that in each iteration where $\|\theta^t-\theta_{opt}\| \geq \gamma r_{opt}$ we get that
    \begin{align*}    \|\theta^{t+1} - \theta_{opt}\|^2&\leq 
    (1 - \frac{2\gamma^2}{8}+\frac{\gamma^4}{64})\|\theta^{t} - \theta_{opt}\|^2 
           + 2(\frac{\gamma^2}{8} - \frac{\gamma^4}{64})\cdot\frac{3}{4} \|\theta - \theta_{opt}\|^2  + \frac{3\gamma^4r_{opt}^2}{64}
           \cr &< (1 - \frac{\gamma^2}{16}) \|\theta^t-\theta_{opt}\|^2 +  \frac{3\gamma^2}{64} \|\theta^t-\theta_{opt}\|^2 = (1-\frac{\gamma^2}{64})\|\theta^t-\theta_{opt}\|^2 
    \end{align*}
    suggesting that after $16T=\frac{64}{\gamma^2}\ln(100/\gamma^2)$ iteration at most it must hold that 
    \[ \|\theta^{16T}-\theta_{opt}\|^2 \leq  \exp(-\frac{64}{\gamma^2}\ln(100/\gamma^2)\cdot \frac{\gamma^2}{64})\|\theta_0-\theta_{opt}\|^2 \leq \frac{\gamma^2}{100}\cdot 100r^2_{opt}  = \gamma^2 r_{opt}^2\]
    %implying that $\|\theta^{16T}-\theta_{opt}\|\leq \gamma r$.
    As required.
    Similarly, if at some iteration $t$ it holds that $\|\theta^t-\theta_{opt}\| < \gamma r_{opt}$ then we get that
    \begin{align*}   
    \|\theta^{t+1} - \theta_{opt}\|^2&\leq 
    (1 - \frac{\gamma^2}{8})^2\gamma^2 r_{opt}^2 
           + 2(\frac{\gamma^2}{8} - \frac{\gamma^4}{64})\cdot\frac{3}{4} \gamma^2 r_{opt}^2  + \frac{3\gamma^4r_{opt^2}}{64}
           \cr &\leq \gamma^2 r_{opt}^2\left( 1 - \frac{2\gamma^2}{8} + \frac{\gamma^4}{64} + \frac{3\gamma^2}{2\cdot 8} - \frac{3\gamma^4}{2\cdot 64} + \frac{3\gamma^2}{64}\right) \leq (1-\frac{\gamma^2}{64})\gamma^2r_{opt}^2
    \end{align*}
    suggesting yet again that $\|\theta^\tau-\theta_{opt}\|<\gamma r_{opt}$ for all $\tau \geq t$.
\end{proof}
}

\cut{
\section{New Results}

\begin{claim}
    Let $\Delta_s \sim \mathcal{N}(\Bar{0}, \sigma_s^2 I_d)$ and $\Delta_c \sim \mathcal{N}(0, \sigma_c^2)$. Let $n_w^t \leftarrow |\{x \in P: x \notin B(\theta^t, r) \}|$. Denote the distribution of $\frac{1}{n_w^t + \Delta_c} \left( \sum\limits_{x \notin B(\theta^t, r)}{(x - \theta^t)} + \Delta_{s} \right)$ as $\mathcal{D}$. Then $\mathcal{D}$ applies:
    \begin{enumerate}
    \item $\innerproduct{\theta_{opt} - \theta^t}{\E\limits_{z \sim \mathcal{D}} [z]} \geq \frac{1}{4}\|\theta^t - \theta_{opt}\|^2$
    \item $\E\limits_{z \sim \mathcal{D}}[\|z\|^2] \leq 512r_{opt}^2$
    \end{enumerate}
\end{claim}

\begin{proof}
\noindent
    \begin{enumerate}
        \item 
        \begin{align*}
            \innerproduct{\theta_{opt} - \theta^t}{\E\limits_{z \sim \mathcal{D}} [z]} &= \innerproduct{\theta_{opt} - \theta^t}{\E\left[\frac{1}{n_w^t + \Delta_c} \right] \sum\limits_{x \notin B(\theta^t, r)}{(x - \theta^t)}}
            \cr & = \innerproduct{\theta_{opt} - \theta^t}{\E\left[\frac{n_w^t}{n_w^t + \Delta_c} \right] \frac{\sum\limits_{x \notin B(\theta^t, r)}{(x - \theta^t)}}{n_w^t}}
            \cr & = \E\left[\frac{1}{1 + \nicefrac{\Delta_c}{n_w^t}} \right]\innerproduct{\theta_{opt} - \theta^t}{ \frac{\sum\limits_{x \notin B(\theta^t, r)}{x}}{n_w^t} - \theta^t}
            \cr & \stackrel{\rm Claim~\ref{clm:pt_outside_ball_dotproduct}}\geq \E\left[\frac{1}{1 + \nicefrac{\Delta_c}{n_w^t}} \right]\cdot\frac{1}{2}\|\theta^t - \theta_{opt}\|^2
            \cr & \stackrel{|\Delta_c| \leq O(\nicefrac{n_w^t}{\sqrt{d}})}\geq \frac{1}{4}\|\theta^t - \theta_{opt}\|^2
        \end{align*}
        
        \item
        \begin{align*}
            \E\limits_{z \sim \mathcal{D}}[\|z\|^2] &= \E\left[\left\|\frac{\sum\limits_{x \notin B(\theta^t, r)}{(x - \theta^t) + \Delta_s}}{n_w^t + \Delta_c}\right\|^2\right]
            \cr & = \E\left[\innerproduct{\frac{\sum\limits_{x \notin B(\theta^t, r)}{(x - \theta^t)}}{n_w^t + \Delta_c} + \frac{\Delta_s}{n_w^t + \Delta_c}}{\frac{\sum\limits_{x \notin B(\theta^t, r)}{(x - \theta^t)}}{n_w^t + \Delta_c} + \frac{\Delta_s}{n_w^t + \Delta_c}}\right]
            \cr & = \E\left[\left(\frac{1}{1 + \nicefrac{\Delta_c}{n_w^t}}\right)^2\right]\left\|\frac{\sum\limits_{x \notin B(\theta^t, r)}{x}}{n_w^t} - \theta^t\right\|^2 + 0 + 0 + \E\left[\frac{\|\Delta_s\|^2}{(n_w^t + \Delta_c)^2}\right]
            \cr & \leq \E\left[\left(\frac{1}{1 + \nicefrac{\Delta_c}{n_w^t}}\right)^2\right](11r_{opt})^2 + \E\left[\frac{\|\Delta_s\|^2}{(n_w^t + \Delta_c)^2}\right]
            \cr & \stackrel{\rm under~ {\cal E}_1}\leq 4\cdot(11r_{opt})^2 + \E[\|\Delta_s\|^2]\cdot\frac{4}{(n_w^t)^2} 
            \cr & \stackrel{\rm under~ {\cal E}_2}\leq 4\cdot(11r_{opt})^2 + \frac{RT\cdot (88r)^2}{\rho}\left(\sqrt{d}+\sqrt{2\ln(\nicefrac{4RT}{\beta_0})}\right)^2\cdot\frac{4}{(n_w^t)^2} 
            \cr & \stackrel{(*)}\leq 484r_{opt}^2 + 4\cdot r_{opt}^2 = 488r_{opt}^2 < 512r_{opt}^2
        \end{align*}
    \end{enumerate}
    ($\ast$) since $ (\tilde n_w^t)^2 \geq \frac{RT \cdot 88^2}{\rho}\left(\sqrt{d}+\sqrt{2\ln(\nicefrac{4RT}{\beta_0})}\right)^2$ in order for us to make an update.
\end{proof}

}

\section{A Differentially Private fPTAS for the MEB Problem}
\label{sec:DP_MEB_fPTAS}

We now turn our attention to the privacy-preserving versions of Algorithms~\ref{alg:np-meb} and~\ref{alg:gd-meb}. In this section we give their curator-model $\rho$-zCDP versions (Algorithms~\ref{alg:dp-meb} and~\ref{alg:ngd-meb} resp.), whereas in the following section (Section~\ref{sec:local_DP_MEB_fPTAS}) we detail their local-model zCDP versions. 
%Due to space constraints, (i) the proof that Algorithm~\ref{alg:dp-meb} is $\rho$-zCDP, (ii) the the full proofs of the following two statements, and (iii) the application of Algorithm~\ref{alg:dp-meb} to Subsample-and-Aggregate are all deferred to the Supplementary Material, Section~\ref{apx_sec:proofsDPAlg}.

\begin{algorithm}[ht]
    \caption{Differentially Private Minimum Enclosing Ball (DP-MEB)}
    \label{alg:dp-meb}
     \hspace*{\algorithmicindent}\textbf{Input:} a set of $n$ points $P \subseteq \mathbb{R}^d$, an approximation parameter $\gamma \in (0,1)$, \\ \hspace*{\algorithmicindent} an initial radius $r_0$ s.t. $r_{opt}\leq r_0 \leq 4r_{opt}$, and an initial center $\theta_0$ s.t. $\|\theta_0 - \theta_{opt}\| \leq 10r_{opt}$, error parameter $\beta$ and privacy-parameter $\rho$.
    \begin{algorithmic}[1]
     % Input:
     % Output:
     \State Remove any $x\in P$ which doesn't belong to $B(\theta_0, 11r_0)$.
     \State Set $i_{\min} \gets 0$, $i_{\max} \gets \ln_{1+\gamma}(4)(\approx \frac{4}{\gamma})$, and $\theta^*\gets\theta_0$.
     \State Set $B \gets \lceil \log_2\left(\ln_{1+\gamma}(4)\right)\rceil$.
     \While{($i_{\min}<i_{\max}$)}
        \State $i_{cur} = \lfloor\frac{i_{\min}+i_{\max}}2\rfloor$
        \State $r_{cur} \leftarrow  (1+\gamma)^{i_{cur}}\cdot r_0 / 4$
        \State $\theta_{cur} \gets {\hyperref[alg:ngd-meb]{\textrm{DP-MMEB}}}(P, \gamma, \frac{\beta}{B}, \frac{\rho}{B} , r_{cur}, \theta_0)$
        \If {($\theta_{cur}\neq \bot$)} 
             \State Set $i_{\max}\gets i_{cur}$, $\theta^*\gets \theta_{cur}$ and $r^* \gets (1+\gamma)r_{cur}$
             \Else
             \State $i_{\min} \gets i_{cur}+1$
        \EndIf
     \EndWhile
     \State \Return $B(\theta^*, r^*)$
   \end{algorithmic}
\end{algorithm}

\newlength\myindent
\setlength\myindent{5mm}
\newcommand\bindent{%
  \begingroup
  \setlength{\itemindent}{\myindent}
  \addtolength{\algorithmicindent}{\myindent}
}
\newcommand\eindent{\endgroup}

\begin{algorithm}[ht]
    \caption{DP-Margin based Minimum Enclosing Ball (DP-MMEB)}
    \label{alg:ngd-meb}
     \hspace*{\algorithmicindent}\textbf{Input:} a set of $n$ points $P \subseteq \mathbb{R}^d$, an approximation parameter $\gamma \in (0,1)$, \\ \hspace*{\algorithmicindent} an error parameter $\beta \in (0,1)$, privacy parameter $\rho$, \\ \hspace*{\algorithmicindent} a candidate radius $r$, and an initial center $\theta_0$ s.t. $\|\theta_0 - \theta_{opt}\| \leq 10r_{opt}$.
    \begin{algorithmic}[1]
     % Input:
     % Output:
     \State Set $R\gets \lceil\log_{\nicefrac 8 7}(\nicefrac 1 \beta) \rceil$,  $\theta^0\gets \theta_0$, $T \leftarrow \frac{4096}{\gamma^2}\ln(\frac{484}{\gamma^2})$, $\beta_0 = \frac 1 {16RT}$, 
      $\sigma_{count}^2 \leftarrow \frac{R(T+1)}{\rho}$,
      and $\sigma_{sum}^2 \leftarrow \frac{RT\cdot (88r)^2}{\rho}$.
    \Repeat {
     \For{$t=0,1,2, \ldots, T-1$}
        \State Sample $\Delta_{count}\sim \mathcal{N}(0, \sigma_{count}^2)$.
        \State $\Tilde{n}^t_{w} \leftarrow |\{x \in P: x \notin B(\theta^t, r) \}| + \Delta_{count}$
        \If{\big($\tilde n_w^t < \frac{88\sqrt{RT}}{\sqrt\rho}\left(\sqrt{d}+\sqrt{2\ln(\nicefrac{4RT}{\beta_0})}\right)$\big)}
            \State \Return $\theta^t$
        \EndIf
        \State Sample $\Delta_{sum} \sim \mathcal{N}(0, \sigma_{sum}^2I_d)$.
        \State Set $\Tilde{\mu}^t_{w} \leftarrow \frac{1}{\Tilde{n}_{w^t}}
        \left(\sum\limits_{x \notin B(\theta^t, r)}{(x - \theta^t)} + \Delta_{sum}\right)$.
        \State Update $\theta^{t+1} \leftarrow \theta^t + \frac{\gamma^2}{2048}\Tilde{\mu}_{w^t}$
     \EndFor
     \State Sample $\Delta_{count}\sim {\cal N}(0, \sigma^2_{count})$.
     \If{\big(      $|P\setminus B(\theta^T,(1+\gamma)r)\}| + \Delta_{count}\leq \sqrt{\frac{2R(T+1)\log(4R(T+1)/\beta_0)}{\rho}}$ \big)} 
     {\Return $\theta^T$ {and halt  }} {}
     \EndIf}
    \Until {$R$ repetitions}
    \State \Return $\bot$
    \end{algorithmic}
\end{algorithm}

\DeclareRobustCommand{\DPanalysis}{
\subsection{Privacy Analysis}
\label{subsec:DP_alg_privacy_analysis}
\begin{lemma}\label{lem:iterative_alg_is_dp}
    \hyperref[alg:ngd-meb]{Algorithm 4} satisfies \emph{$\rho$-zCDP}.
\end{lemma}

\begin{proof}
    At each one of the $RT$ iterations of the algorithm, we answer two queries to the input data: a counting query and a summation query.
    It is known that the \emph{$L_2$-sensitivity} of a counting query is $1$, therefore using the Gaussian mechanism theorem while setting $\sigma_{count}^2 = \frac{R(T+1)}{\rho}$ satisfies \emph{$\frac{\rho}{2R(T+1)}$-zCDP}. 
    Secondly, we know that all the points are bounded by a ball of radius $11r_0 \leq 44r_{opt}\leq 44r$ around $\theta_0$, hence the summation query has $L_2$-sensitivity of $\leq 88r$. Thus, by setting $\sigma_{sum}^2 = \frac{RT(88r)^2}{\rho}$ we have that we answer each summation query using $\frac{\rho}{2T}$-zCDP.
    Due to sequential composition of zCDP~\cite{bun16}, it holds that in all $T$ iteration together we preserve $\left(\rho(1-\frac{1}{2R(T+1)})\right)$-zCDP. Lastly, we apply one last counting query which we answer using the Gaussian mechanism while satisfying $\frac{\rho}{2R(T+1)}$-zCDP, thus, overall we are $\rho$-zCDP.
\end{proof}

\begin{corollary}\label{cor:overall_alg_is_DP}
    \hyperref[alg:dp-meb]{Algorithm 3} satisfies \emph{$\rho$-zCDP}.
\end{corollary}
\begin{proof}
    Since Algorithm~\ref{alg:dp-meb} invokes $B=\lceil \log_2(\log_{1+\gamma}(4)) \rceil$ calls to Algorithm~\ref{alg:ngd-meb} each preserving $\frac{\rho}B$-zCDP, Algorithm~\ref{alg:dp-meb} is $\rho$-zCDP overall.
\end{proof}
}

\DPanalysis

\subsection{Utility Analysis}
\label{subsec:utility_analysis}

\DeclareRobustCommand{\lemDPUtility}
{W.p. $\geq1-\beta$, applying Algorithm~\ref{alg:ngd-meb} with $r\geq r_{opt}$ and an initial center $\theta_0$ s.t. $\|\theta_0-\theta_{opt}\|\leq 10r_{opt}$ returns a point $\theta^t$ where $\left| P\setminus B(\theta^t, (1+\gamma)r) \right| \leq  88\sqrt{\frac{RT}{\rho}}\left(\sqrt{d}+\sqrt{2\ln(\nicefrac{4RT}{\beta_0})}\right)+ \sqrt{\frac{2R(T+1)\log(4R(T+1)/\beta_0)}{\rho}}$.}

\begin{lemma}
\label{lem:DP_iterative_algorithm_utility}
\lemDPUtility
\end{lemma}
\DeclareRobustCommand{\lemDPUtilityPf}{
\begin{proof}
    Given a repetition $r$ and iteration $t$ denote the events
    \begin{align*}
        {\cal E}_1^{r,t} &:= \text{in the $(r,t)$-draws, $|\Delta_{count}|\leq \sigma_{count}\sqrt{2\ln(\nicefrac{4R(T+1)}{\beta_0})}$}
        \cr {\cal E}_2^{r,t} &:= \text{in the $(r,t)$-draw, $\|\Delta_{sum}\|\leq \sigma_{sum}\left(\sqrt{d}+\sqrt{2\ln(\nicefrac{4RT}{\beta_0})}\right)$}
    \end{align*}
    and denote also ${\cal E}_i = \bigcup_{r,t}{\cal E}_i^{r,t}$ for $i=1,2$. Using standard bounds on the concentration of the Gaussian distribution and the $\chi^2_d$-distribution together with the union-bound we have that $\Pr[\overline{{\cal E}_1}\cup\overline{{\cal E}_2}]\leq R(T+1)\cdot \frac{\beta_0}{2R(T+1)}+RT\frac{\beta_0}{2RT}\leq \beta_0$. We continue the rest of the proof conditioning on ${\cal E}_1\cap {\cal E}_2$ holding.

    Fix $r$ and $t$. Under ${\cal E}_1^{r,t}\cap {\cal E}_2^{r,t}$ holding, the required conditions detailed in~\eqref{eq:requirements_of_good_dist} hold, which~-- using Corollary~\ref{cor:random_step_algorithm_whp}~-- yields the correctness of our algorithm. Under the same notation as in Algorithm~\ref{alg:ngd-meb}, 
    denote the distribution of $\frac{1}{n_w^t + \Delta_{count}} \left( \sum\limits_{x \notin B(\theta^t, r)}{(x - \theta^t)} + \Delta_{sum} \right)$ as $\mathcal{D}^t$. 
    
    First, observe that under ${\cal E}_1^{r,t}$, the condition $\tilde n_w^t  \geq \frac{88\sqrt{RT}}{\sqrt\rho}\left(\sqrt{d}+\sqrt{2\ln(\nicefrac{4RT}{\beta_0})}\right)$ implies that 
    \[ n_w^t \geq  
    \frac{88\sqrt{RT}}{\sqrt\rho}\left(\sqrt{d}+\sqrt{2\ln(\nicefrac{4RT}{\beta_0})}\right)-\sqrt{\frac{2R(T+1)\ln(\nicefrac{4R(T+1)}{\beta_0})}{\rho}}\geq 44 |\Delta_{count}|
    \]
    and secondly, observe that $\Delta_{sum}$ is drawn from a spherically symmetric distribution, so for any $a>0$ we have that $\E[\Delta_{sum}|~ \|\Delta_{sum}\|\leq a]=0$.
    And so, if indeed Algorithm~\ref{alg:ngd-meb} passes the if-condition and makes an update step we have
    \begin{align*}
        \E\limits_{z \sim \mathcal{D}^t} \left[\innerproduct{\theta_{opt} - \theta^t}{z}|~\theta^t, {\cal E}_1^{r,t}\cap {\cal E}_2^{r,t}\right] &=
        \innerproduct{\theta_{opt} - \theta^t}{\E\left[\frac{ \Delta_{sum}+\sum\limits_{x \notin B(\theta^t, r)}{(x - \theta^t)}}{n_w^t + \Delta_{count}} |~\theta^t, {\cal E}_1^{r,t}\cap {\cal E}_2^{r,t}\right]}
            \cr & \stackrel{\rm independ.}= \innerproduct{\theta_{opt} - \theta^t}{\E\left[\frac{1}{n_w^t + \Delta_{count}} |~\theta^t, {\cal E}_1^{r,t}\right] {\sum\limits_{x \notin B(\theta^t, r)}{(x - \theta^t)}}}
            \cr &= \innerproduct{\theta_{opt} - \theta^t}{\E\left[\frac{n_w^t}{n_w^t + \Delta_{count}} |~\theta^t, {\cal E}_1^{r,t}\right] \frac{\sum\limits_{x \notin B(\theta^t, r)}{(x - \theta^t)}}{n_w^t}}
            \cr & = \E\left[\frac{1}{1 + \nicefrac{\Delta_{count}}{n_w^t}}|~\theta^t, {\cal E}_1^{r,t} \right]\innerproduct{\theta_{opt} - \theta^t}{ \frac{\sum\limits_{x \notin B(\theta^t, r)}{x}}{n_w^t} - \theta^t}
            \cr & \stackrel{\rm Claim~\ref{clm:pt_outside_ball_dotproduct}}\geq \left(\frac{1}{1 - \nicefrac{1}{44}} \right)\cdot\frac{1}{2}\|\theta^t - \theta_{opt}\|^2\geq \frac{1}{4}\|\theta^t - \theta_{opt}\|^2
    \end{align*}
    and also
    \begin{align*}
            &\E\limits_{z \sim \mathcal{D}^t}[\|z\|^2|~\theta^t, {\cal E}_1^{r,t}\cap {\cal E}_2^{r,t}] = \E\left[\left\|\frac{\sum\limits_{x \notin B(\theta^t, r)}{(x - \theta^t) }}{n_w^t + \Delta_{count}}+\frac{\Delta_{sum}}{n_w^t + \Delta_{count}}\right\|^2|~\theta^t, {\cal E}_1^{r,t}\cap {\cal E}_2^{r,t}\right]
            %\cr & = \E\left[\innerproduct{\frac{\sum\limits_{x \notin B(\theta^t, r)}{(x - \theta^t)}}{n_w^t + \Delta_c} + \frac{\Delta_s}{n_w^t + \Delta_c}}{\frac{\sum\limits_{x \notin B(\theta^t, r)}{(x - \theta^t)}}{n_w^t + \Delta_c} + \frac{\Delta_s}{n_w^t + \Delta_c}}\right]
            \cr &~~~ = \resizebox*{0.93\textwidth}{!}{$\E\left[\left(\frac{n_w^t}{n_w^t + \Delta_{count}}\right)^2\left\|\frac{\sum\limits_{x \notin B(\theta^t, r)}{x}}{n_w^t} - \theta^t\right\|^2 + \frac{2\langle \Delta_{sum},\sum\limits_{x \notin B(\theta^t, r)}{(x - \theta^t) }\rangle+\|\Delta_{sum}\|^2}{(n_w^t + \Delta_{count})^2}  |~\theta^t, {\cal E}_1^{r,t}\cap {\cal E}_2^{r,t}\right]$}
            \cr &~~~ \stackrel{\rm independ.}= \E\left[\left(\frac{1}{1 + \frac{\Delta_{count}}{n_w^t}}\right)^2|~\theta^t, {\cal E}_1^{r,t}\right]\left\|\frac{\sum\limits_{x \notin B(\theta^t, r)}{x}}{n_w^t} - \theta^t\right\|^2 + \frac {0+\E\left[{\|\Delta_{sum}\|^2}|~\theta^t, {\cal E}_2^{r,t}\right]} {(\tilde n_w^t)^2}
            \cr &~~~ \leq \frac 1 {1-\nicefrac 1{44}}\cdot(11r)^2 + \frac{RT\cdot (88r)^2}{\rho}\left(\sqrt{d}+\sqrt{2\ln(\nicefrac{4RT}{\beta_0})}\right)^2\cdot\frac{1}{(\tilde n_w^t)^2}
            < 512r^2
        \end{align*}
    since $\tilde n_w^t \geq 88\sqrt{\frac{RT}{\rho}}\left(\sqrt{d}+\sqrt{2\ln(\nicefrac{4RT}{\beta_0})}\right)$ in order for us to make an update.

    Corollary~\ref{cor:random_step_algorithm_whp} suggests that if we make all $T$ updates then indeed $\|\theta^T-\theta_0\|\leq \gamma r$ and so $|P\setminus B(\theta^T, (1+\gamma)R)|=0$. So under ${\cal E}_1$ Algorithm~\ref{alg:ngd-meb} returns $\theta^T$. Otherwise, at some iteration we do not make an update step, which under ${\cal E}_1$ suggests that 
    \[ n_w^t = |P\setminus B(\theta^t,r)| \leq 88\sqrt{\frac{RT}{\rho}}\left(\sqrt{d}+\sqrt{2\ln(\nicefrac{4RT}{\beta_0})}\right)+ \sqrt{\frac{2R(T+1)\log(4R(T+1)/\beta_0)}{\rho}} \qedhere \]
\end{proof}
}
\lemDPUtilityPf
\DeclareRobustCommand{\corDPUtility}
{Given $r_0$ where $r_{opt}\leq r_0\leq 4r_{opt}$ and a point $\theta_0$ where $\|\theta_0-\theta^*\|\leq 10r_{opt}$, w.p. $\geq 1-\beta$ Algorithm~\ref{alg:dp-meb} is a $O(n\cdot \frac{\log^2(\nicefrac 1 \gamma)\log(\nicefrac 1 \beta)}{\gamma^2} )$-time algorithm that returns a ball $B(\theta^*,r)$ where
    $r \leq (1+3\gamma)r_{opt}$ and where 
    $|P\setminus B(\theta^*,r^*)|= O( \frac{\left(\sqrt{d}+\sqrt{\log(\nicefrac{\log(\nicefrac 1 \beta)}{\gamma})}\right)\sqrt{\log(\nicefrac 1 \gamma)\log(\nicefrac 1 \beta)}}{\gamma\sqrt{\rho}}  )$.}
\begin{corollary}
\label{cor:DP_utility}
\corDPUtility
\end{corollary}
\DeclareRobustCommand{\corDPUtilityPf}{
\begin{proof}
    The result follows directly from the fact that Algorithm~\ref{alg:dp-meb} invokes $B = O(\log(\nicefrac 1 \gamma))$ calls to Algorithm~\ref{alg:ngd-meb}, with a privacy budget of $O(\rho/\log(\nicefrac 1 \gamma))$ each and with a failure probability of $O(\beta/\log(\nicefrac 1 \gamma))$ each. Plugging those into the bound of Lemma~\ref{lem:DP_iterative_algorithm_utility} together with the fact that $T = O(\gamma^{-2}\log(\nicefrac 1 \gamma))$ yields the resulting bound. 
    Note that, denoting the ``correct'' $i^* = \min\{i\geq 0: ~ \frac{r_0}4 (1+\gamma)^i\geq r_{opt}\}$, under the event that no invocation of Algorithm~\ref{alg:ngd-meb} fails, each time we execute the binary search with a value of $i_{cur}\geq i^*$ we obtain some $\theta_{cur}\neq \bot$. Due to the nature of the binary search and the fact that upon finding $\theta_{cur}\neq \bot$ we set $i_{\max}=i_{cur}$, it must follows that we return a ball of radius $(1+\gamma)r^*=(1+\gamma)\cdot \frac{r_0}4\cdot (1+\gamma)^i$ for some $i\leq i^*$, and so $r^*\leq (1+\gamma)^2r_{opt}\leq (1+3\gamma)r_{opt}$.
    Lastly, the runtime of Algorithm~\ref{alg:ngd-meb} is $O(nRT)$ making the runtime of Algorithm~\ref{alg:dp-meb} to be $O(nRTB)=O(\frac{n\log^2(1/\gamma)\log(1/\beta)}{\gamma^2})$ as required.
\end{proof}
}
\corDPUtilityPf

We comment that the amplification of the success probability of the algorithm from $\nicefrac 1 8$ to $1-\beta$ can be done using the amplification techniques of~\cite{LiuT19} which saves on the privacy budget: instead of na\"ively setting the privacy budget per iteration as $\rho/R$, we could use conversions to $(\epsilon,\delta)$-DP and as a result ``shave-off'' a factor of $R$. But since $R=O(\log(1/\beta))$ this would merely reduce polyloglog factors, at the expense of readability.

\DeclareRobustCommand{\algDPApplication}{
\subsection{Application: Subsample Stable Functions}
\label{subsec:subsample_stable_functions}

Much like the work of~\cite{GhaziKM20}, our work too is applicable as a DP-aggregator in a Subsample-and-Aggregate~\cite{nissim2007smooth} framework. We say that a point $p\in \R^d$ is \emph{$(r,\beta)$-stable} for some function $f:\mathcal{X}^* \to \R^d$  if there exists $m(r,\beta)$ such that for any input $S\subset {\mathcal{X}}^n$ a random subsample of $m$ entries of $S$ input datapoints returns w.p. $\geq 1-\beta$ a value close to $p$, namely, $\Pr_{S'\subset S, |S|=m}[\|c-f(S')\|\leq r]\geq 1-\beta$.

\newcommand{\T}{\mathsf{T}}
\begin{theorem}
    Fix $\rho,\gamma,\beta>0$. There exists some constant $C>0$ such that the following holds. Suppose $f:\mathcal{X}^* \to \R^d$ is a function that has a $(r,\beta)$-stable point. Then, there exists a $\rho$-zCDP algorithm that takes an input a dataset $S\subset {\cal X}^n$ and w.p.$\geq 1- \beta$ returns a $((1+\gamma)r,\nicefrac{\beta}{2k})$-stable point provided that $n \geq k\cdot m(r,\nicefrac{\beta}{2k})$ for $k=\frac{C\left(\sqrt{d}+\sqrt{\log(\nicefrac{\log(\nicefrac 1 \beta)}{\gamma})}\right)\sqrt{\log(\nicefrac 1 \gamma)\log(\nicefrac 1 \beta)}}{\gamma\sqrt{\rho}} $. Furthermore, if finding $f(S')$ for any $S'$ containing $m(r,\nicefrac{\beta}{2k})$-many datapoint takes $\T$ time, then our algorithm runs in time $O(k\T + k\cdot \frac{\log^2(\nicefrac 1 \gamma)\log(\nicefrac 1 \beta)}{\gamma^2})$.
\end{theorem}
\begin{proof}
    The proof simply partitions the $n$ inputs points of $S$ into $k$ disjoint and random subsets $S'_1, S'_2,..., S'_k$. W.p. $\geq 1- \beta/2$ it holds that $\|f(S'_i)-c\|\leq r$ for every subset $S_i'$, and then we apply our $(1+\gamma)$ approximation over this dataset of $k$ many points (with a failure probability of $\beta/2$) and returns the resulting center-point.
\end{proof}
This results improves on Theorem 18 of~\cite{GhaziKM20} in both the runtime and the required number of subsamples, at the expense of requiring \emph{all} subsamples to be close to the point $p$ rather than just many of the points.
}
\algDPApplication

\section{A Local-DP fPTAS for the MEB Problem}
\label{sec:local_DP_MEB_fPTAS}

In this section we give the local-model version of our algorithm. At the core of its utility proof is a lemma analogous to Lemma~\ref{lem:DP_iterative_algorithm_utility}, in which we prove that w.h.p.~in each iteration $t$ the distribution of our update-step satisfies (w.h.p.) the requirements of~\eqref{eq:requirements_of_good_dist}. 

\cut{
Again, due to space constraints, we merely state the LDP algorithm in this section, and defer both its privacy and utility analyses to the Supplementary Material, Section~\ref{apx_sec:proofsLDPAlg}~--- where we prove that it is a $\rho$-zCDP algorithm that returns a ball $B(\theta^*,r^*)$ such that $r^*\leq (1+3\gamma)r_{opt}$ and $|P\setminus B(\theta^*,r^*)| = O(\frac{\sqrt{n}\log(\nicefrac 1 \gamma)}{\gamma^2\sqrt{\rho}}\left(\sqrt{d}+\sqrt{\log(\nicefrac{1}{\gamma\beta})}\right))$.
}

\begin{algorithm}[ht]
    \caption{LDP-Margin based Minimum Enclosing Ball (LDP-MMEB)}
    \label{alg:LDP-meb}
     \hspace*{\algorithmicindent}\textbf{Input:} a set of $n$ points $P \subseteq \mathbb{R}^d$, an approximation parameter $\gamma \in (0,1)$, \\ \hspace*{\algorithmicindent} an error parameter $\beta \in (0,1)$, privacy parameter $\rho$, \\ \hspace*{\algorithmicindent} a candidate radius $r$, and an initial center $\theta_0$ s.t. $\|\theta_0 - \theta_{opt}\| \leq 10r_{opt}$.
    \begin{algorithmic}[1]
     % Input:
     % Output:
     \State Set $R\gets \lceil\log_{\nicefrac 8 7}(\nicefrac 1 \beta) \rceil$, $\theta^0\gets \theta_0$, $T \leftarrow \frac{4096}{\gamma^2}\ln(\frac{484}{\gamma^2})$, $\beta_0 = \frac 1 {16RT}$, 
      $\sigma_{count}^2 \leftarrow \frac{R(T+1)}{\rho}$,
      and $\sigma_{sum}^2 \leftarrow \frac{RT\cdot (88r)^2}{\rho}$.
     \Repeat{
     \For{$t=0,1,2, \ldots, T-1$}
        \For {\textbf{each }$(x\in P)$}
        \State Sample $\Delta_{count}\sim \mathcal{N}(0, \sigma_{count}^2)$.
        \State Sample $\Delta_{sum} \sim \mathcal{N}(0, \sigma_{sum}^2I_d)$.
        \If {($x \notin B(\theta^t,r)$)}
        {\State Send $Y^t_x = 1 + \Delta_{count}, Z^t_x = x-\theta^t + \Delta_{sum}$
        \Else\  Send $Y^t_x = \Delta_{count}, Z^t_x = \Delta_{sum}$}
        \EndIf\EndFor
        \State Set $\tilde n^t_w = \sum_{x\in P}Y^t_x$ and $\tilde v_w^t = \frac 1 {\tilde n^t_w}\sum_{x\in P} Z^t_x$. 
        \If {\big($\tilde n_w^t < \frac{88\sqrt{nRT}}{\sqrt\rho}\left(\sqrt{d}+\sqrt{2\ln(\nicefrac{4RT}{\beta_0})}\right)$\big)} {\Return $\theta^t$} \EndIf
        \State Update $\theta^{t+1} \leftarrow \theta^t +  \frac{\gamma^2}{2048}\tilde v_w^t$
     \EndFor
     \For {\textbf{each }$(x\in P)$}
        \State Sample $\Delta_{count}\sim \mathcal{N}(0, \sigma_{count}^2)$.
        \If {($x \notin B(\theta^T,(1+\gamma)r)$)}
        {\State Send $Y^T_x = 1 + \Delta_{count}$
        \Else\  Send $Y^T_x = \Delta_{count}$}
        \EndIf\EndFor
    \State Set $n_w^T \gets \sum_{x}Y_x$
     \If{\big(   $n_w^T  \leq \sqrt{\frac{2nR(T+1)\log(\nicefrac{4R(T+1)}{\beta_0})}{\rho}}$  \big)} 
     {\Return $\theta^T$ and halt}
     \EndIf}
     \Until {$R$ repetitions}
     \State \Return $\bot$
    \end{algorithmic}
\end{algorithm}

\cut{
The crux of its proof is as follows. For each iteration $t$ we denote $n_w^t$ as the true number of datapoints in $P$ outside the ball $B(\theta^t, r)$, $\mu_w^t$ as their true mean, and $v_w^t$ as the difference of the true mean and the current center $v_w^t = \mu_w^t-\theta^t$; and denote their LDP analogues as $\tilde n_w^t$ and $\tilde v_w^t$ (as defined in Algorithm~\ref{alg:LDP-meb}) and $\tilde \mu_w^t =\tilde v_w^t +\theta^t$. We then define the events ${\cal E}_1 := \text{in all iterations, $|\tilde n_w^t - n_w^t|\leq \sqrt{\frac{2n(T+1)\log(\nicefrac{4(T+1)}{\beta})}{\rho}}$}$ and ${\cal E}_2 := \text{in all iterations, $\|\sum_{x}Z_x^t - n_w^t v_w^t\|\leq \frac{88r\sqrt{nT}\left(\sqrt{d}+\sqrt{2\ln(\nicefrac{4T}{\beta})}\right)}{\sqrt{\rho}}$}$
which hold together w.p. $\geq 1-\beta$. (The r.v. $Z_x^t$ is defined in Algorithm~\ref{alg:LDP-meb}.) Under ${\cal E}_1$ and ${\cal E}_2$ we have that
    \begin{align*}
        \|\tilde\mu_w^t-\mu_w^t\| &= 
        \|\tilde v_w^t- v_w^t\|
        \leq \left\| \frac{\sum_{x}Z_x^t-n_w^t v_w^t}{\tilde n_w^t}- \left(n_w^t v_w^t\right) \left( \frac{1}{\tilde n_w^t} - \frac{1}{n_w^t} \right)  \right\| \cr &\leq \frac{\|\sum_{x}Z_x^t-n_w^t v_w^t\|}{\tilde n_w^t} + \|\mu_w^t-\theta^t\| \frac{|\tilde n_w^t-n_w^t|}{ \tilde n_w^t}
        \cr & \leq \frac {88r\sqrt{nT}\left(\sqrt{d}+\sqrt{2\ln(\nicefrac{4T}{\beta})}\right)+ 22r\sqrt{2n(T+1)\log(\nicefrac{4(T+1)}{\beta})}}{\tilde n_w^t\sqrt{\rho}} \leq \frac{\gamma r}{16}
    \end{align*}
    as Algorithm~\ref{alg:LDP-meb} only makes an update step if $\tilde n_w^t \geq \frac{16\sqrt{n(T+1)}}{\gamma\sqrt{\rho}}\left(88\sqrt{d}+ 110\sqrt{2\log(\nicefrac{4(T+1)}{\beta})}\right)$.
}
\DeclareRobustCommand{\LDPSection}{
\begin{claim}
Algorithm~\ref{alg:LDP-meb} is a local-model $\rho$-zCDP.
\end{claim}
\begin{proof}
    The proof is very similar to the proof of Lemma~\ref{lem:iterative_alg_is_dp}~--- where we apply basically the same accounting, noticing that each $x\in P$ is in charge of randomizing her own data, making this algorithm LDP.
\end{proof}

\begin{lemma}
\label{lem:LDP_alg_utility}
W.p. $\geq1-\beta$, applying Algorithm~\ref{alg:ngd-meb} with $r\geq r_{opt}$ and an initial center $\theta_0$ s.t. $\|\theta_0-\theta_{opt}\|\leq 10r_{opt}$ returns a point $\theta^t$ where $\left| P\setminus B(\theta^t, (1+\gamma)r) \right| \leq  88\sqrt{\frac{nRT}{\rho}}\left(\sqrt{d}+\sqrt{2\ln(\nicefrac{4nRT}{\beta_0})}\right)+ \sqrt{\frac{2R(T+1)\log(4R(T+1)/\beta_0)}{\rho}}$.
\end{lemma}
\begin{proof}
    Analogously to the proof of Lemma~\ref{lem:DP_iterative_algorithm_utility}, we use the similar definitions: in each iteration $t$ we denote $n_w^t$ as the true number of datapoints in $P$ outside the ball $n_w^t = |\{x\in P:~ x\notin B(\theta^t, r)\}|$,\footnote{Where technically, in the last steps of the algorithm, $n_w^T = |\{x\in P:~ x\notin B(\theta^T, (1+\gamma)r)\}|$.} $\mu_w^t$ as their true mean $\mu_W^t = \frac{1}{n_w^t}\sum_{x\notin B(\theta^t,r)}x$, and $v_w^t$ as the difference of the true mean and the current center $v_w^t = \mu_w^t-\theta^t = \frac{1}{n_w^t}\sum_{x\notin B(\theta^t,r)}(x-\theta^t)$. We thus define the events
    \begin{align*}
        {\cal E}_1 &:= \text{in all $T+1$ iterations, $|\tilde n_w^t - n_w^t|\leq \sqrt{\frac{2nR(T+1)\log(\nicefrac{4(T+1)}{\beta})}{\rho}}$}
        \cr {\cal E}_2 &:= \text{in all $T$ iterations, $\|\sum_{x}Z_x^t - n_w^t v_w^t\|\leq \frac{88r\sqrt{nRT}}{\sqrt{\rho}}\left(\sqrt{d}+\sqrt{2\ln(\nicefrac{4T}{\beta})}\right)$}
    \end{align*}
    Proving that both $\Pr[\overline{\cal E}_1]\leq \beta/2$ and $\Pr[\overline{\cal E}_2]\leq \beta/2$ is rather straight-forward. In each iteration $t$ it holds that $\sum_x{Y^t_x}\sim{\cal N}(n_w^t, n\sigma^2_{count})$ as the sum on $n$ independent Gaussians, and so we merely apply standard Gaussian concentration  bounds together with the union bound over all $T+1$ iterations. 
    Similarly, in each iteration $t$ it holds that $\sum_x Z^t_x \sim {\cal N}(n_w^t(\mu_x^t-\theta^t), n\sigma^2_{sum}I_d)$. So standard bounds on the concentration of the $\chi^2_d$-distribution assert that the $L_2$-distance between the random draw from such a $d$-dimensional Gaussian and its mean is $> \sqrt{n\sigma^2_{sum}}(\sqrt{d} + \sqrt{2\ln(\nicefrac{4T}{\beta})}$ w.p. $<\frac{\beta}{2T}$, after which we apply the union-bound on all $T$ iterations. We continue the rest of the proof conditioning on both ${\cal E}_1$ and ${\cal E}_2$  holding.

    Again, due to our if-condition, we make an update-step only when $\tilde n_w^t$ is large, which, under ${\cal E}_1$ implies that
    \[ n_w^t \geq  
    \frac{88\sqrt{nRT}}{\sqrt\rho}\left(\sqrt{d}+\sqrt{2\ln(\nicefrac{4RT}{\beta_0})}\right)-\sqrt{\frac{2nR(T+1)\ln(\nicefrac{4R(T+1)}{\beta_0})}{\rho}}\geq 44 |\Delta_{count}|
    \]
    and then proving that the distribution which we use to make an update-step satisfies the conditions detailed in~\eqref{eq:requirements_of_good_dist} w.h.p.~is precisely the same proof (using the independence of $\Delta_{count}$ and $\Delta_{sum}$ and the fact that $\E[\Delta_{sum}]=0$).

    Invoking Corollary~\ref{cor:random_step_algorithm_whp} we have that if we make all $T$ updates then indeed $\|\theta^T-\theta_0\|\leq \gamma r$ and so $|P\setminus B(\theta^T, (1+\gamma)R)|=0$. So under ${\cal E}_1$ Algorithm~\ref{alg:ngd-meb} returns $\theta^T$. Otherwise, at some iteration we do not make an update step, which under ${\cal E}_1$ suggest that 
    \[ n_w^t = |P\setminus B(\theta^t,r)| \leq 88\sqrt{\frac{nRT}{\rho}}\left(\sqrt{d}+\sqrt{2\ln(\nicefrac{4nRT}{\beta_0})}\right)+ \sqrt{\frac{2R(T+1)\log(4R(T+1)/\beta_0)}{\rho}} \qedhere \]
\end{proof}

\begin{corollary}
    \label{cor:LDP_utility} 
    Algorithm~\ref{alg:dp-meb} altered so it invokes $B=O(\log(\nicefrac 1 \gamma))$ calls to Algorithm~\ref{alg:LDP-meb} (instead of Algorithm~\ref{alg:ngd-meb}) is a $O(\frac {\log(\nicefrac 1\beta)\log^2(\nicefrac 1 \gamma)}{\gamma^{2}})$-rounds $\rho$-zCDP algorithm in the local-model that takes $O(n\cdot \frac{\log^2(\nicefrac 1 \gamma)\log(\nicefrac 1 \beta)}{\gamma^2} )$-time; and that returns a ball $B(\theta^*,r^*)$ such that $r^*\leq (1+3\gamma)r_{opt}$ and $|P\setminus B(\theta^*,r^*)| = O( \frac{\sqrt n\left(\sqrt{d}+\sqrt{\log(\nicefrac{\log(\nicefrac 1 \beta)}{\gamma})}\right)\sqrt{\log(\nicefrac 1 \gamma)\log(\nicefrac 1 \beta)}}{\gamma\sqrt{\rho}}  )$.
\end{corollary}
\begin{proof}
    The proof follows from using the bound of Lemma~\ref{lem:LDP_alg_utility} with $T=O(\gamma^{-2}\log(\nicefrac 1 \gamma))$, and with a privacy budget of $\rho/B$ and failure probability of $\beta/B$ in each invocation of Algorithm~\ref{alg:LDP-meb}.
\end{proof}
}
\LDPSection

\DeclareRobustCommand{\ExperimentsSection}{\input{experiment}}
\input{experiment}

%\vspace{-2mm}
\section{Discussion and Open Problems}
\label{sec:conclusion}
This work is the first to give a DP-fPATS for the MEB problem, in both the curator- and the local-model, and it leads to numerous open problems. The first is the question of improving the utility guarantee. Specifically, the number of points our algorithm may omit from $P$ has a dependency of $\tilde O(\nicefrac 1 \gamma)$ in the approximation factor, where this dependency follows from the fact that in each of our $T=\tilde O(\gamma^{-2})$ iterations. Thus finding either an iterative algorithm which makes $\ll T$ iterations or a variant of SVT that will allow the privacy budget to scale like $O(\log(T))$ will reduce this dependency to only ${\rm polylog}(\gamma^{-1})$. Alternatively, it is intriguing whether there exists a lower-bound for any zCDP PTAS of the MEB problem proving a polynomial dependency on $\gamma$. (The best we were able to prove is via packing argument~\cite{HardtT10, BunS16} using a grid of $O((\nicefrac 1 \gamma)^d)$ many points, leading to a $d\log(\nicefrac 1 \gamma)$ bound.)

A different open problem lies on the the application of this DP-MEB approximation to the task of DP-clustering, and in particular~--- on improving on the works of~\cite{HuangL18,ShechnerSS20,CohenKMST21} for ``stable'' $k$-median/means clustering. One can presumably combine our technique with the LSH-based approach used in~\cite{NissimStemmer18} to cover a subset of points lying close together, however~--- it is unclear to us what is the effect of using only some of each cluster's ``core'' on the approximated MEB we return and on the $k$-means/median cost. More importantly, it does not seem that for the $k$-means problem our MEB approximation yields a better cost than the simple baseline of DP-averaging each cluster's core (after first finding a $O(1)$-MEB approximation, as discussed in the introduction). But it is possible that our work can be a building block in a first PTAS for the $k$-center problem in low-dimensions, a setting in which the $k$-center problem has a non-private PTAS~\cite{Har-Peled11}.

\begin{ack}
O.S. is supported by the BIU Center for
Research in Applied Cryptography and Cyber Security in
conjunction with the Israel National Cyber Bureau in the
Prime Minister’s Office, and by ISF grant no. 2559/20. Both
authors thank the anonymous reviewers for terrific suggestions and advice on improving this paper.
\end{ack}
\bibliographystyle{plain}
\bibliography{bibliography,paper}

\cut{
\section*{Checklist}

\begin{enumerate}

\item For all authors...
\begin{enumerate}
  \item Do the main claims made in the abstract and introduction accurately reflect the paper's contributions and scope?
    \answerYes{}
  \item Did you describe the limitations of your work?
    \answerYes{See lines 103-105} 
  \item Did you discuss any potential negative societal impacts of your work?
    \answerNA{}
  \item Have you read the ethics review guidelines and ensured that your paper conforms to them?
    \answerYes{}
\end{enumerate}

\item If you are including theoretical results...
\begin{enumerate}
  \item Did you state the full set of assumptions of all theoretical results?
    \answerYes{}
        \item Did you include complete proofs of all theoretical results?
    \answerYes{See also Supplementary Material Sections~\ref{apx_sec:proofsDPAlg} and~\ref{apx_sec:proofsLDPAlg}}
\end{enumerate}

\item If you ran experiments...
\begin{enumerate}
  \item Did you include the code, data, and instructions needed to reproduce the main experimental results (either in the supplemental material or as a URL)?
    \answerYes{In the Supplementary Material}
  \item Did you specify all the training details (e.g., data splits, hyperparameters, how they were chosen)?
    \answerNA{}
        \item Did you report error bars (e.g., with respect to the random seed after running experiments multiple times)?
    \answerNA{}
        \item Did you include the total amount of compute and the type of resources used (e.g., type of GPUs, internal cluster, or cloud provider)?
    \answerYes{}
\end{enumerate}

\item If you are using existing assets (e.g., code, data, models) or curating/releasing new assets...
\begin{enumerate}
  \item If your work uses existing assets, did you cite the creators?
    \answerYes{}
  \item Did you mention the license of the assets?
    \answerYes{See under `Bar Crawl' in the experiment section~\ref{sec:experiments}.}
  \item Did you include any new assets either in the supplemental material or as a URL?
    \answerYes{As a URL, see References}
  \item Did you discuss whether and how consent was obtained from people whose data you're using/curating?
    \answerNA{}
  \item Did you discuss whether the data you are using/curating contains personally identifiable information or offensive content?
    \answerNA{}
\end{enumerate}

\item If you used crowdsourcing or conducted research with human subjects...
\begin{enumerate}
  \item Did you include the full text of instructions given to participants and screenshots, if applicable?
    \answerNA{}
  \item Did you describe any potential participant risks, with links to Institutional Review Board (IRB) approvals, if applicable?
    \answerNA{}
  \item Did you include the estimated hourly wage paid to participants and the total amount spent on participant compensation?
    \answerNA{}
\end{enumerate}

\end{enumerate}
}

\newpage

\appendix

\section{Finding an Initial Good Center}
\label{apx_sec:preliminary_algorithms}

In this section we give, for completeness, the $\rho$-zCDP version of the algorithms for approximating $P$'s optimal radius up to a constant factor and finding some $\theta_0$ which is sufficiently close to the center of $P$'s MEB. The algorithm itself is ridiculously simple, and has appeared before implicitly. We bring it here for two reasons: (a) completeness and (b) in its LDP-version, this algorithm's utility depends solely on $\sqrt n$. Thus, combining this algorithm with the Algorithm~\ref{alg:LDP-meb} of Section~\ref{sec:local_DP_MEB_fPTAS}, we obtain a LDP-fPTAS for the MEB problem who's utility depends on $\sqrt{n}$ rather than the $n^{0.67}$-bound of \cite{NissimStemmer18} (at the expense of worse dependency on other parameters). This gives a clear improvement on previous algorithms for approximating the MEB problem when $n\to \infty$. Our algorithm requires a starting point $\theta_0$ which is $R_{\max}$ away from all points in $P$ (namely, $P\subset B(\theta_0, R_{\max})$, and a lower bound $r_{\min}$ on $r_{opt}$; and its overall utility bounds depends on $\log(R_{\max}/r_{\min})$. In a standard setting, where $P\subset [-B,B]^d$ and where all points lie on some grid ${\cal G}^d$ whose step-size is $\tau$, we can set $\theta_0$ as the origin and set $R_{\max}= B\sqrt d$ and $r_{\min}=\tau/2$, resulting in $O(\log(\nicefrac{Bd}{\tau}))$-dependency. In the specific case where $r_{opt}=0$ and all datapoints in $P$ lie on the exact same grid point we can just return the closest grid point to the resulting $\theta$ once it get to a radius of $r=r_{min}=\tau/2$.

\begin{algorithm}[H]
    \caption{Noisy Average and Radius (GoodCenter)}
    \label{alg:good-center}
     \hspace*{\algorithmicindent}\textbf{Input:} a set of $n$ points $P$ and parameters $\theta_0,R_{\max}$ and $r_{\min}$, such that $P\subset B(\theta_0, R_{\max})$ and $r_{opt}\geq r_{\min}$. Failure parameter $\beta \in (0,1)$, privacy parameter $\rho$.
    \begin{algorithmic}[1]
     % Input:
     % Output:
     \State Set $T\gets \lceil\log_2(R_{\max}/r_{\min})\rceil+1$, $X\gets \sqrt{\frac{2T\ln(\nicefrac{4T}{\beta})}{\rho}}$
     \State Set $\sigma_{count}^2 \leftarrow \frac{T}{\rho}$, $\sigma_{sum}^2 \leftarrow \frac{T}{\rho}$.
     \State Init $P^0\gets P$, $\theta^0\gets \theta_0$, $n_{cur}\gets n$ and $r_{cur}\gets R_{\max}$.
     \For {($t=0,1,2,..., T-1$)}
     \State $P^t \gets P^t \cap B(\theta^t,r_{cur})$.
     \State $\Delta_{sum} \sim \mathcal{N}(0, 4r_{cur}^2\sigma_{sum}^2I_d)$
     \State $\tilde \mu^t \gets ({\sum_{x\in P^t}x + \Delta_{sum}})/{n_{cur}}$
     \State $\Delta_{count}\gets {\cal N}(0, \sigma^2_{count})$
     \If {($|P^t\setminus B(\tilde\mu^t, \frac 1 2 r_{cur})|+\Delta_{count} \geq  X$)}
            {\Return $B(\theta^t, r_{cur})$}\EndIf
    \State Update: $r_{cur}\gets \frac 1 2 r_{cur}$, $n_{cur}\gets n_{cur}-2X$, $\theta^{t+1}\gets \tilde \mu^t$. 
    \EndFor
     \State \Return $B(\theta^T, r_{cur})$
    \end{algorithmic}
\end{algorithm}

\begin{theorem}
    \label{thm:AlgIterativeCentering_DP}
    Algorithm~\ref{alg:good-center} is $\rho$-zCDP.
\end{theorem}
\begin{proof}
    The proof follows immediately from the fact that the $L_2$-global sensitivity of a count query is 1, and that the $L_2$-global sensitivity of a sum of datapoints in a ball of radius $r_{cur}$ is at most $2r_{cur}$. The rest of the proof relies on the composition of $2T$ queries, each answered with a ``budget'' of $\frac{\rho}{2T}$-zCDP.
\end{proof}

\begin{theorem}
    \label{thm:good_center_alg_utility}
    W.p. $\geq 1-\beta$, given a set of points $P$ of size $n$ where $n\geq \max\{  16T\sqrt{\frac{2T\ln(\nicefrac{4T}{\beta})}{\rho}},  16\sqrt{\frac{T}{\rho}}(\sqrt{d}+{\sqrt{2\ln(\nicefrac{4T}{\beta})}})\}$, Algorithm~\ref{alg:good-center} returns a ball $B(\theta^*, r^*)$ where (i) the set $P' = P\cap B(\theta^*,r^*)$ contains at least $n-\sqrt{\frac{8T^3\ln(\nicefrac{4T}{\beta})}{\rho}}$, and (ii) denoting $B(\theta(P'), r_{opt}(P'))$ as the MEB of $P'$, we have that $r^*\leq 6r_{opt}$.
    %and   $\|\theta^*-\theta(P')\|\leq 4r*$.
\end{theorem}
\begin{proof}
    Let ${\cal E}$ be the event where for any of the $\leq T$ draws of the $\Delta_{sum}$ and $\Delta_{count}$ it holds that
    \[  |\Delta_{count}| \leq \sqrt{\frac{2T\ln(\nicefrac{4T}{\beta})}{\rho}} \qquad \text{and} \qquad \|\Delta_{sum}\| \leq 2r_{cur}\sqrt{\frac{T}{\rho}}(\sqrt{d}+{\sqrt{2\ln(\nicefrac{4T}{\beta})}}) \]
    where again, standard union bound and Gaussian / $\chi^2$-distribution concentration bounds give that $\Pr[\overline{\cal E}]\leq \beta$. So we continue the proof under the assumption that ${\cal E}$ holds.
    
    In this case, in any iteration it must hold that $|P\setminus B(\mu^t, \frac 1 2 r_{cur})|\leq 2X = \sqrt{\frac{8T\ln(\nicefrac{4T}{\beta})}{\rho}}$. It follows that all in all we remove in the process of Algorithm~\ref{alg:good-center} at most $2XT$ points, and since $n\geq 16XT$ we have that in any iteration $t$ it always holds that $n\geq |P^t|\geq n-2Xt = n_{cur}\geq \frac{7n}{8}\geq 14XT$. Denoting in any iteration $t$ the true mean of the points (remaining) in $P^t$ as $\mu_t = \frac{1}{|P^t|}\sum_{x\in P^t} x$, and the center of the MED of $P^t$ as $\theta_{t}$, it follows that
    \begin{align*}
        \|\tilde\mu^t-\mu^t\| &= \|\tilde\mu^t-\theta_t-(\mu^t-\theta_t)\| =\left\| \frac{\Delta_{sum}+ \sum_{x\in P^t}(x-\theta_t)}{n_{cur}} - \frac{\sum_{x\in P^t}(x-\theta_t)}{|P^t|} \right\| 
        \cr &\leq \left\| \frac{\Delta_{sum}}{n_{cur}}\right\| + \left\| \frac{\left(\sum_{x\in P^t} (x-\theta_t)\right)(|P^t|-n_{cur})}{|P^t|n_{cur}}     \right\| \leq  \frac{8\|\Delta_{sum}\|}{7n}+\|\mu^t-\theta_t\| \frac{2XT}{n_{cur}}
       \cr &\leq \frac {8\cdot 2r_{cur}\sqrt{\frac{T}{\rho}}(\sqrt{d}+{\sqrt{2\ln(\nicefrac{4T}{\beta})}})}{7n}
        + \frac{r_{opt}(P^t)}{7}\leq \frac{r_{cur}+r_{opt}(P^t)}{7}
    \end{align*}
    Since we assume $n\geq16\sqrt{\frac{T}{\rho}}(\sqrt{d}+{\sqrt{2\ln(\nicefrac{4T}{\beta})}})$. Moreover, since $\|\mu^t-\theta_t\|\leq r_{opt}(P^t)$ it follows that $\|\tilde\mu^t - \theta_t\|\leq \frac{r_{cur}+8r_{opt}(P^t)}{7}$. Now, as long as $r_{cur}\geq 6r_{opt}(P^t)$ we have that
    \begin{align*}
        \frac {r_{cur}}{2} \geq \frac{r_{cur}}{7} + \frac{5r_{cur}}{14} \geq \frac{r_{cur}}{7} + \frac{30r_{opt}(P^t)}{14} \geq r_{opt}(P^t)+\frac{r_{cur}+8r_{opt}(P^t)}{7} \geq r_{opt}(P^t)+\|\tilde \mu^t - \theta_t\|
    \end{align*}
    thus $B({\theta_t, r_{opt}(P^t)})\subset B(\tilde\mu^t, \frac{r_{cur}}{2})$ which implies that $|P^t\setminus B(\mu^t, \frac 1 2 r_{cur})|=0$, and so under ${\cal E}$ we continue to the next iteration. 
    
    And so, when we halt it must hold that $r_{cur}$ (which is the $r^*$ we return) must satisfy that $r_{cur}<6r_{opt}(P^t)$. 
    %This implies that when we halt $\|\tilde\mu^t - \theta_t\|\leq \frac{r_{cur}+8r_{opt}(P^t)}{7} \leq \frac{16}{7} r_{opt}(P^t)$.
\end{proof}
\begin{corollary}
    \label{cor:good_center}
    Algorithm~\ref{alg:good-center} is a $\rho$-zCDP algorithm that, given $n$ points on a grid ${\cal G}\subset [-B,B]^d$ of side-step $\tau$ where $n=\Omega(\sqrt{\frac{\log(\nicefrac{Bd}\tau)}{\rho}}(\sqrt d + \sqrt{\log(\nicefrac{Bd}{\tau\beta})}))$ returns w.p. $\geq 1-\beta$ a ball $B(\theta^*,r^*)$ where for $P'=P\setminus B(\theta^*,r^*)$ it holds that both $n-|P'|=O(\frac{\log(\nicefrac{Bd}\tau)}{\sqrt\rho}\sqrt{\log(\nicefrac{Bd}{\tau\beta})}))$ and that w.r.t to $B(\theta_{opt},r_{opt})$ which is the true MEB of $P'$ we have that $\|\theta^*-\theta_{opt}\|\leq 6r_{opt}(P')$.
\end{corollary}

\subsection{A Local-DP Version of Finding an Initial Good Center}
\label{apx_subsec:LDP_good_center}

\begin{algorithm}[H]
    \caption{LDP Noisy Average and Radius (LDP-GoodCenter)}
    \label{alg:LDP-good-center}
     \hspace*{\algorithmicindent}\textbf{Input:} a set of $n$ points $P$ and some parameter $R_{\max},\theta_0$ and $r_{\min}$, such that $P\subset B(\theta_0, R_{\max})$ and $r_{opt}\geq r_{\min}$. Failure parameter $\beta \in (0,1)$, privacy parameter $\rho$.
    \begin{algorithmic}[1]
     % Input:
     % Output:
     \State Set $T\gets \lceil\log_2(R_{\max}/r_{\min})\rceil+1$, $X\gets \sqrt{\frac{2nT\ln(\nicefrac{4T}{\beta})}{\rho}}$
     \State $\sigma_{count}^2 \leftarrow \frac{T}{\rho}$, $\sigma_{sum}^2 \leftarrow \frac{T}{\rho}$.
     \State Init $\theta^0\gets \theta_0$, and $r_{cur}\gets R_{\max}$.
     \For {($t=0,1,2,..., T-1$)}
     \State Denote $\Pi^t$ as the projection onto $B(\theta^t,r_{cur})$.
     \For {{\bf each} $x\in P$}
     \State Send $Y_x \sim \mathcal{N}(\Pi^t(x), 4r_{cur}^2\sigma_{sum}^2I_d)$
     \EndFor
     \State $\tilde \mu^t \gets \frac 1 n \sum_x Y_x$
     \For {{\bf each} $x\in P$}
     \If { ($x\notin B(\tilde \mu^t, \frac 1 2 r_{cur})$)  } 
        {\State Send $Z_x \sim \mathcal{N}(1, \sigma_{count}^2)$
        \Else\  Send $Z_x \sim \mathcal{N}(0, \sigma_{count}^2)$}
     \EndIf\EndFor
     \If {($\sum_x Z_x \geq  X$)}
            {\Return $B(\theta^t, r_{cur})$}\EndIf
    \State Update: $r_{cur}\gets \frac 1 2 r_{cur}$, $\theta^{t+1}\gets \tilde \mu^t$. 
    \EndFor
     \State \Return $B(\theta^T, r_{cur})$
    \end{algorithmic}
\end{algorithm}

\begin{theorem}
    \label{thm:LDP_good_center_alg_utility}
    Algorithm~\ref{alg:LDP-good-center} is a LDP algorithm in which each user maintains $\rho$-zCDP. Forthermore, w.p. $\geq 1-\beta$, given a set of point $P$ of size $n$ where $n\geq \max\{  16T\sqrt{\frac{2nT\ln(\nicefrac{4T}{\beta})}{\rho}},  16\sqrt{\frac{nT}{\rho}}(\sqrt{d}+{\sqrt{2\ln(\nicefrac{4T}{\beta})}})\}$, Algorithm~\ref{alg:LDP-good-center} returns a ball $B(\theta^*, r^*)$ where the set $P' = \{\Pi_{B(\theta^*,r^*)}(x):~ x\in P\}$ contains no more than $2T\sqrt{\frac{2T\ln(\nicefrac{4T}{\beta})}{\rho}}$ points for which $x\neq \Pi_{B(\theta^*,r^*)}(x)$; and denoting $B(\theta(P'), r_{opt}(P'))$ as the MEB of $P'$, it holds that  $\|\theta^*-\theta(P')\|\leq 8r*$.
\end{theorem}
The proof of Theorem~\ref{thm:LDP_good_center_alg_utility} is completely analogous to the proof of Theorems~\ref{thm:AlgIterativeCentering_DP} and~\ref{thm:good_center_alg_utility} using the fact that in each iteration $t$ of the algorithm
\begin{align*}
    &\sum_{x}Y_x \sim {\cal N}\left( \sum_x \Pi^t(x),~~4nr_{cur}^2\sigma^2_{sum}I_d\right)  
    \cr&\sum_{x}Z_x \sim {\cal N}\left( |\{x\in P:~ x\notin B(\tilde\mu^t, r_{cur}/2)\}|, ~~n\sigma_{count}^2\right)
\end{align*}
\begin{corollary}
    \label{cor:LDP_good_center}
    Algorithm~\ref{alg:LDP-good-center} is a $\rho$-zCDP algorithm in the local-model that, given $n$ points on a grid ${\cal G}\subset [-B,B]^d$ of side-step $\tau$ where $n=\Omega({\frac{\log(\nicefrac{Bd}\tau)}{\rho}}(\sqrt d + \sqrt{\log(\nicefrac{Bd}{\tau\beta})})^2)$ returns w.p. $\geq 1-\beta$ a ball $B(\theta^*,r^*)$ where for the set $P'=\{\Pi_{B(\theta^*,r^*)}(x):~ x\in P\}$ it holds that at most $O(\frac{\sqrt{n}\cdot\log(\nicefrac{Bd}\tau)}{\sqrt\rho}\sqrt{\log(\nicefrac{Bd}{\tau\beta})})$ points are shifted in the projection (and the rest remain as they are in $P$) and that w.r.t to $B(\theta_{opt},r_{opt})$ which is the true MEB of $P'$ we have that $\|\theta^*-\theta_{opt}\|\leq 6r^*$.
    %and $r^*\leq 8r_{opt}$.
\end{corollary}

    Note that comparing Corollary~\ref{cor:LDP_good_center} with the approximation of~\cite{NissimStemmer18}, we have that they may omit $O(n^{0.67}\log(n/\tau))$-many points whereas we may omit  only $\sqrt{n}\log^{3/2}(d/\tau)$ points. But, of course, they deal with a bounding ball for $t$ points out of giving $n$, whereas we deal with the MEB problem.

\section{Using Noisy Mean}
\label{apx_sec:SQNoise}
Here we continue the analysis detailed in Section~\ref{subsec:noisy_version}.
\SQNoiseParagraph

\cut{
\section{Missing Proofs: DP Algorithm}
\label{apx_sec:proofsDPAlg}
}

%\DPanalysis
%\subsection{Utility Analysis and Sample Complexity}
%\label{subsec:DP_alg_utility_analysis}

\cut{
\begin{lemma}
\label{apx_lem:DP_iterative_algorithm_utility}[Lemma~\ref{lem:DP_iterative_algorithm_utility} restated.]
\lemDPUtility
\end{lemma}
}
%\lemDPUtilityPf

\cut{
\begin{corollary}
\label{apx_cor:DP_utility}
    [Corollary~\ref{cor:DP_utility} restated]
    \corDPUtility
\end{corollary}
}
%\corDPUtilityPf

%\algDPApplication

\cut{
\section{Missing Proofs: Local-DP Algorithm}
\label{apx_sec:proofsLDPAlg}
}
%\LDPSection

%\ExperimentsSection

\end{document}

%% file: experiment.tex
\section{Experiments}
\label{sec:experiments}

In this section we give an experimental evaluation of our algorithm on three synthetic datasets and one real dataset. We emphasize that our experiment should be perceived  merely as a proof-of-concept experiment aimed at the possibility of improving the algorithm's analysis, and not a thorough experimentation for a ready-to-deploy code. We briefly explain the experimental setup below.

\paragraph{Goal.}
We set to investigate the performance of our algorithm, and seeing whether the performance is similar across different types of input and across a range of parameters. In addition, we wondered whether in practice our algorithm halts prior to concluding all  $T=O(\gamma^{-2}\ln(\nicefrac 1 \gamma))$ iterations.

\paragraph{Experiment details.}
We conducted experiments solely with Algorithm~\ref{alg:ngd-meb} with update-step that uses a constant learning rate of $\nicefrac{\gamma^2}{8}$, feeding it the true $r_{opt}$ of each given dataset as its $r$ parameter.
By default, we used the following set of parameters. Our domain in the synthetic experiments is $[-5,5]^{10}$ (namely, we work in the $10$-dimensional space), and our starting point $\theta_0$ is the origin. The default values of our privacy parameter is $\rho = 0.3$, of the approximation constant is $1.2$ (namely $\gamma = 0.2$), and of the failure probability is $\beta=e^{-9}\approx 0.00012$. We set the maximal number of repetitions $T$ just as detailed in Algorithm~\ref{alg:ngd-meb}, which depends on $\gamma$. 

We varied two of the input parameters, $\rho$ and $\gamma$, and also the data-type. We ran experiments with $\rho \in \{0.1, 0.3, 0.5, 0.7, 0.9\}$ and with $\gamma \in \{0.1, 0.2, 0.3, 0.4, 0.5\}$. Based on the values of $\rho$ and $\gamma$ we computed $n_0 =\frac{\sqrt{RT}(\sqrt{d}+\sqrt{\ln(\nicefrac {4RT} \beta_0)}}{\sqrt{\rho}}$ which we used as our halting parameter. In all experiments involving a synthetic dataset, we set the input size $n$  to be $n=640n_0$.

We varied also the input type, using 3 synthetically generated datasets and one real-life dataset:
\begin{itemize}
    \item \underline{Spherical Gaussian:} we generated samples from a $d$-dimensional Gaussian $\mathcal{N}(v, I_d)$, where $v \in \mathbb{R}^d$ is a random shift vector. We discarded each point that did not fall in $[-5,5]^{10}$.
    \item \underline{Product Distribution:} we generated samples from a $d$-dimensional Bernoulli distribution with support $\{-1, 1\}^d$ with various probabilities for each dimension~--- where for each coordinate $i\in [10]$ we set $\Pr[x_i = 1]=2^{-i}$. This creates a ``skewed'' distribution whose mean is quite far from its $1$-center. In order for the $1$-center not to coincide with $\theta_0=\bar 0$ we shifted this cube randomly in the grid.
    \item \underline{Conditional Gaussian:} we repeated the experiment with the spherical Gaussian only this time we conditioned our random draws so that no coordinate lies in the $[0,0.5]$-interval. This skews the mean of the distribution to be $<0$ in each coordinate, but leaves the $1$-center unaltered. Again, we shifted the Gaussian to a random point $v\in [-5,5]^d$.   
    \item \underline{``Bar Crawl: Detecting Heavy Drinking'':} a dataset taken from the freely available UCI Machine Learning Repository~\cite{UCI} which collected accelerometer data from participants in a college bar crawl~\cite{KillianPNMC19}. We truncated the data to only its 3 $x$-, $y$- and $z$-coordinates, and dropped any entry outside of $[-1,1]^3$, and since it has two points $(-1,-1,-1)$ and $(1,1,1)$ then its $1$-center is the origin (so we shifted the data randomly in the $[-5,5]^3$ cube). This left us with $n=12,921,593$ points. Note that the data is taken from a very few participants, so our algorithm gives an event-level privacy~\cite{DworkNPR10}.
\end{itemize}
We ran our experiments in Python, on a (fairly standard) Intel Core i7 \@ 2.80 GHz with 16GB RAM and they run in time that ranged from $15$ seconds (for $\gamma=0.5$) to $2$ hours (for $\gamma=0.1$).

\paragraph{Results.}
The results are given in Figures~\ref{fig:results_by_gamma}, \ref{fig:results_by_rho}, where we plotted the distance of $\theta^t$ to $\theta_{opt}$ for each set of parameters across $t=10$ repetitions. As evident, we converged to a good approximation of the MEB in all settings. We halt the experiment (i) if $\|\theta_t - \theta_{opt}\| \leq \gamma r_{opt}$, or (ii) if there are not enough wrong points, or (iii) if $t > 2500$ indicating that the run isn't converging.
Indeed, the number of iterations until convergence does increase as $\gamma$ decreases; but, rather surprisingly, varying $\rho$ has  a small effect on the halting time. This is somewhat expected as $T$ has no dependency on $\rho$ whereas its dependency on $\gamma$ is proportional to $\gamma^{-2}$, but it is evident that as $\rho$ increases our mean-estimation in each iteration becomes more accurate, so one would expect a faster convergence. Also unexpectedly, our results show that even for datasets whose mean and $1$-center aren't close to one another (such as the Conditional Gaussian or Product-Distribution), the number of iterations until convergence remains roughly the same (see for example Figure~\ref{fig:results_by_gamma} vs.~\ref{fig:results_by_rho}).

\begin{figure}	
	\centering
	\begin{subfigure}[t]{0.475\textwidth}
		\centering
		\includegraphics[width=\textwidth]{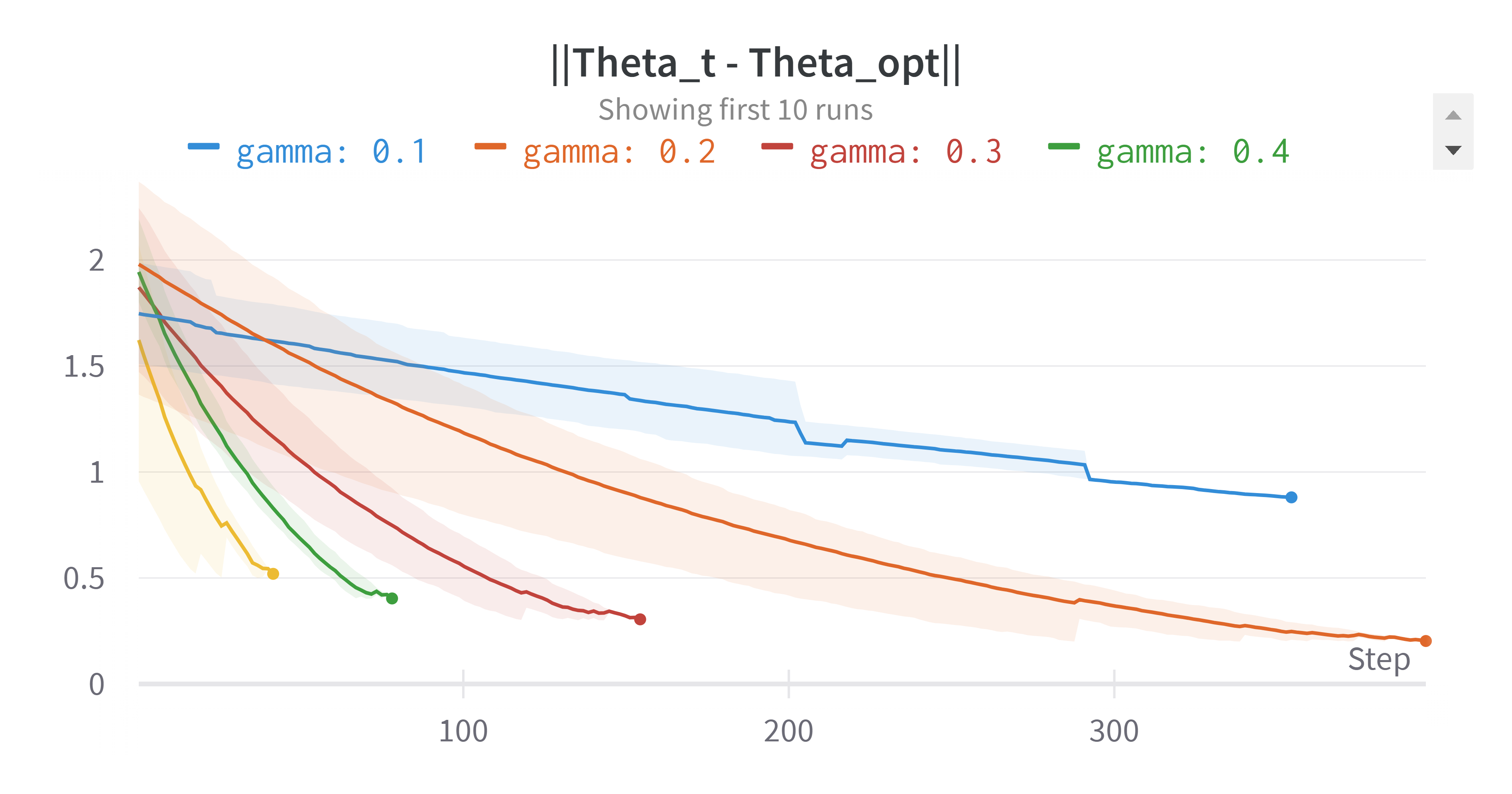}
		\caption{Spherical Gaussian}
	\end{subfigure}
	\hfill
	\begin{subfigure}[t]{0.475\textwidth}
		\centering
		\includegraphics[width=\textwidth]{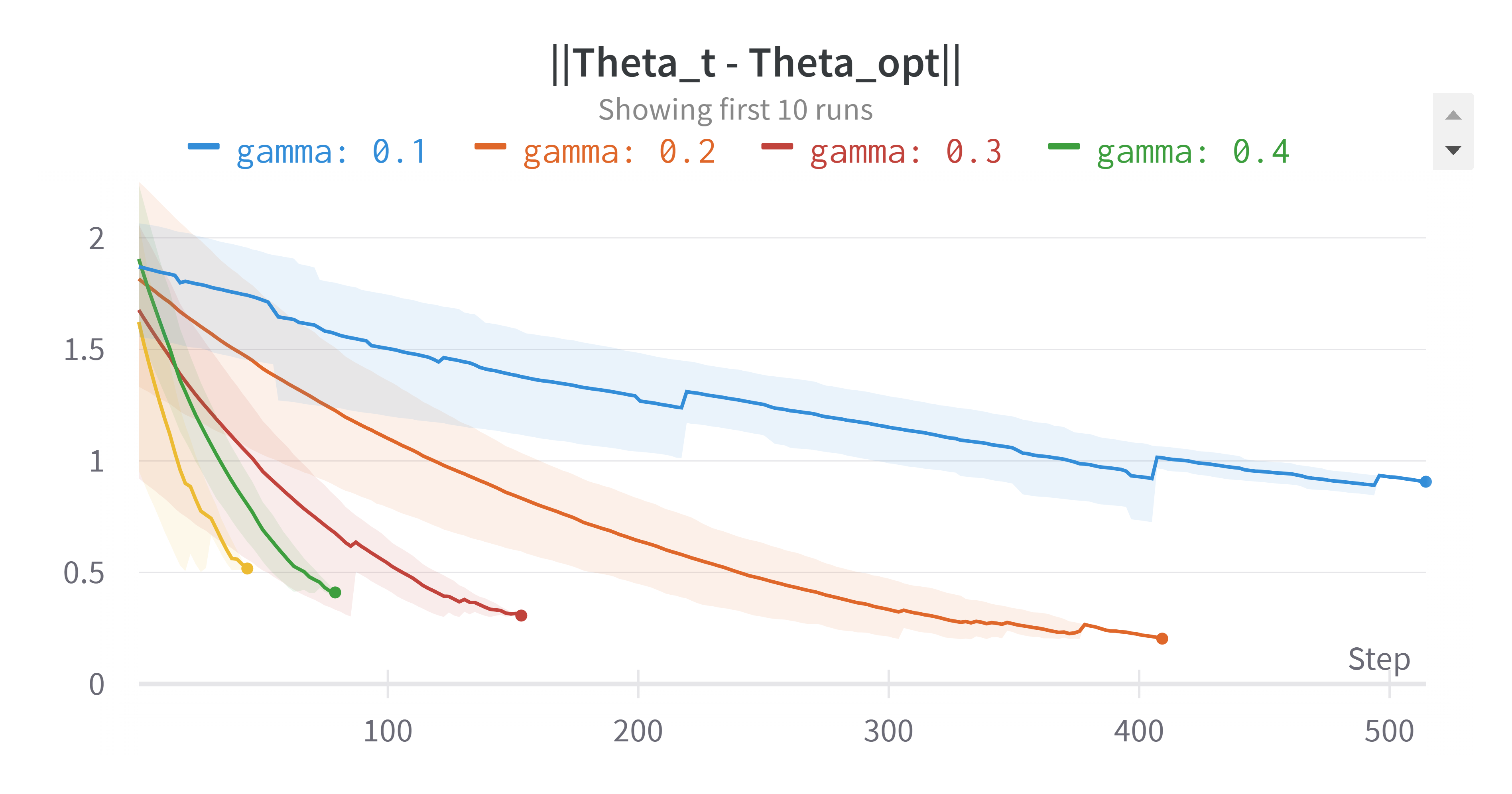}
		\caption{Conditional Gaussian}
	\end{subfigure}
	\vskip\baselineskip
	\begin{subfigure}[t]{0.475\textwidth}
		\centering
		\includegraphics[width=\textwidth]{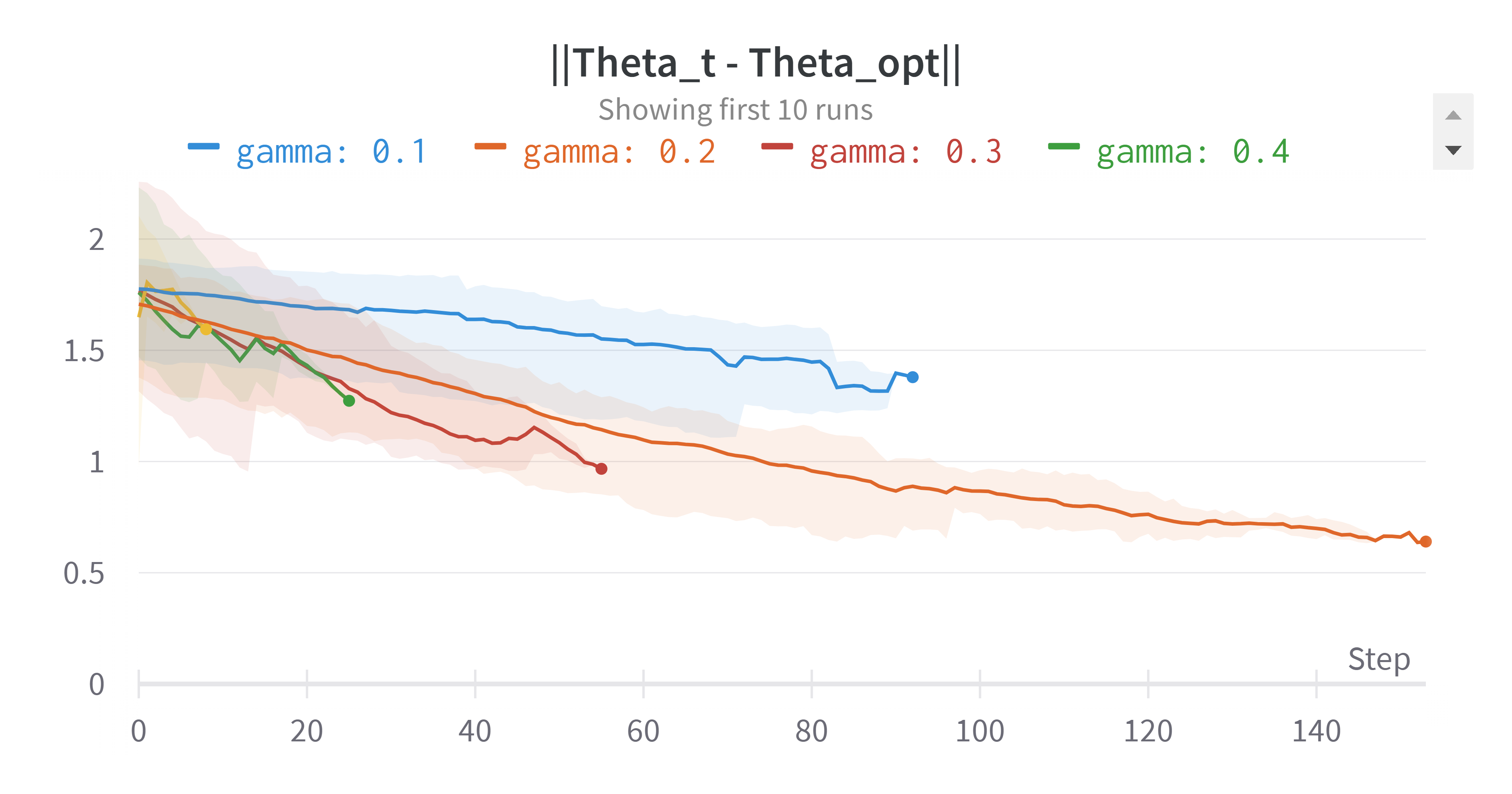}
		\caption{Product Distribution}
	\end{subfigure}
	\hfill
	\begin{subfigure}[t]{0.475\textwidth}
		\centering
		\includegraphics[width=\textwidth]{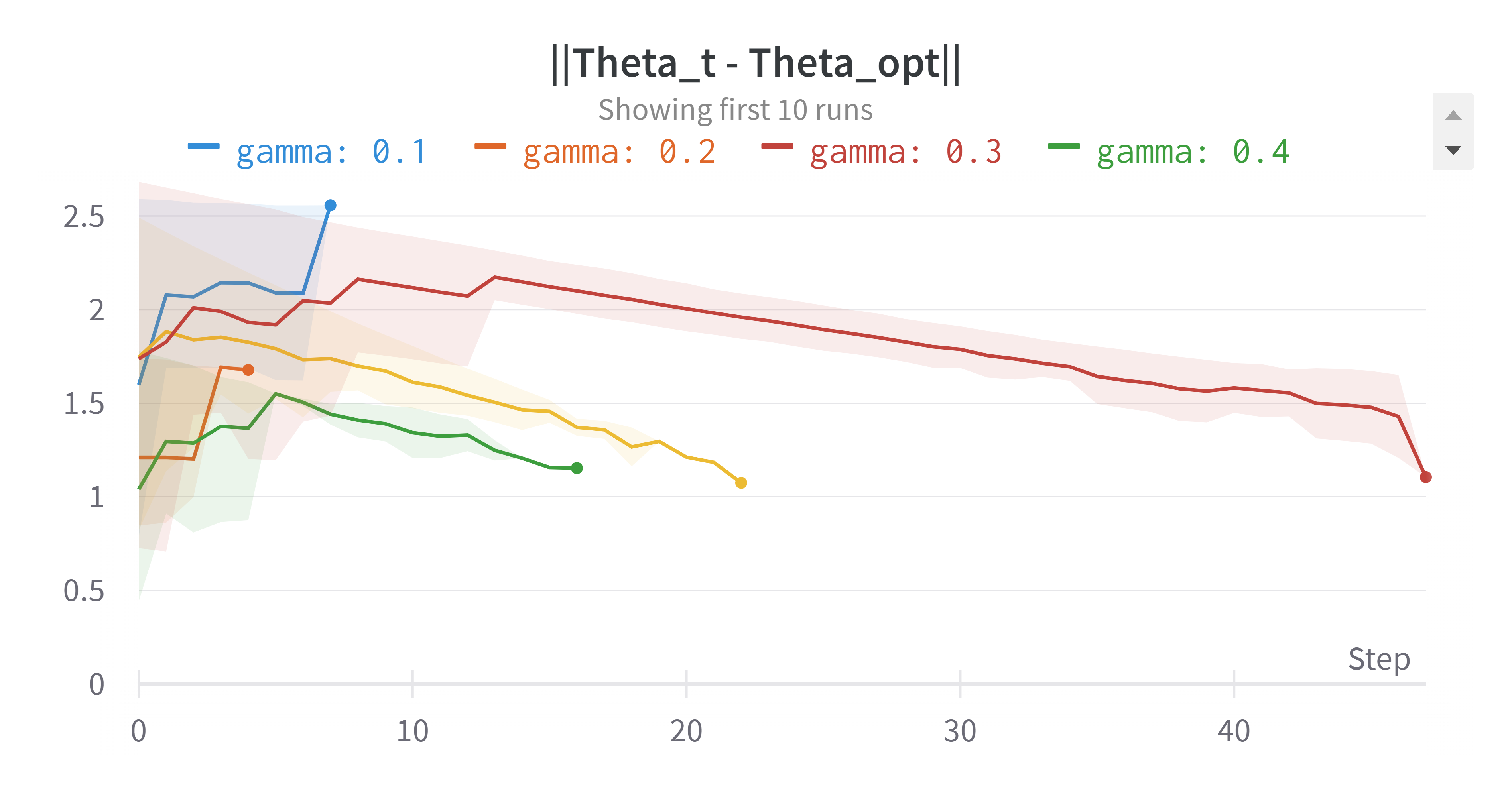}
		\caption{\label{subfig:barcrawl_gamma}Bar Crawl: Detecting Heavy Drinking}
	\end{subfigure}
	\caption{\label{fig:results_by_gamma} The distance of $\theta^t$ to $\theta_{opt}$ as a function of $t$ -- the iteration number, for $\rho = 0.3$ and $\gamma \in \{0.1, 0.2, 0.3, 0.4, 0.5 \}$. Each curve corresponds to a different $\gamma$ value. In all experiments the number of iterations until convergence does increase as $\gamma$ decreases, except for $\gamma=0.1$ where it halts because there were not enough wrong points. Note that for $\gamma=0.1$ for Bar Crawl dataset (figure~\ref{subfig:barcrawl_gamma}) we didn't converge due to its size.}
\end{figure}

\begin{figure}	
	\centering
	\begin{subfigure}[t]{0.475\textwidth}
		\centering
		\includegraphics[width=\textwidth]{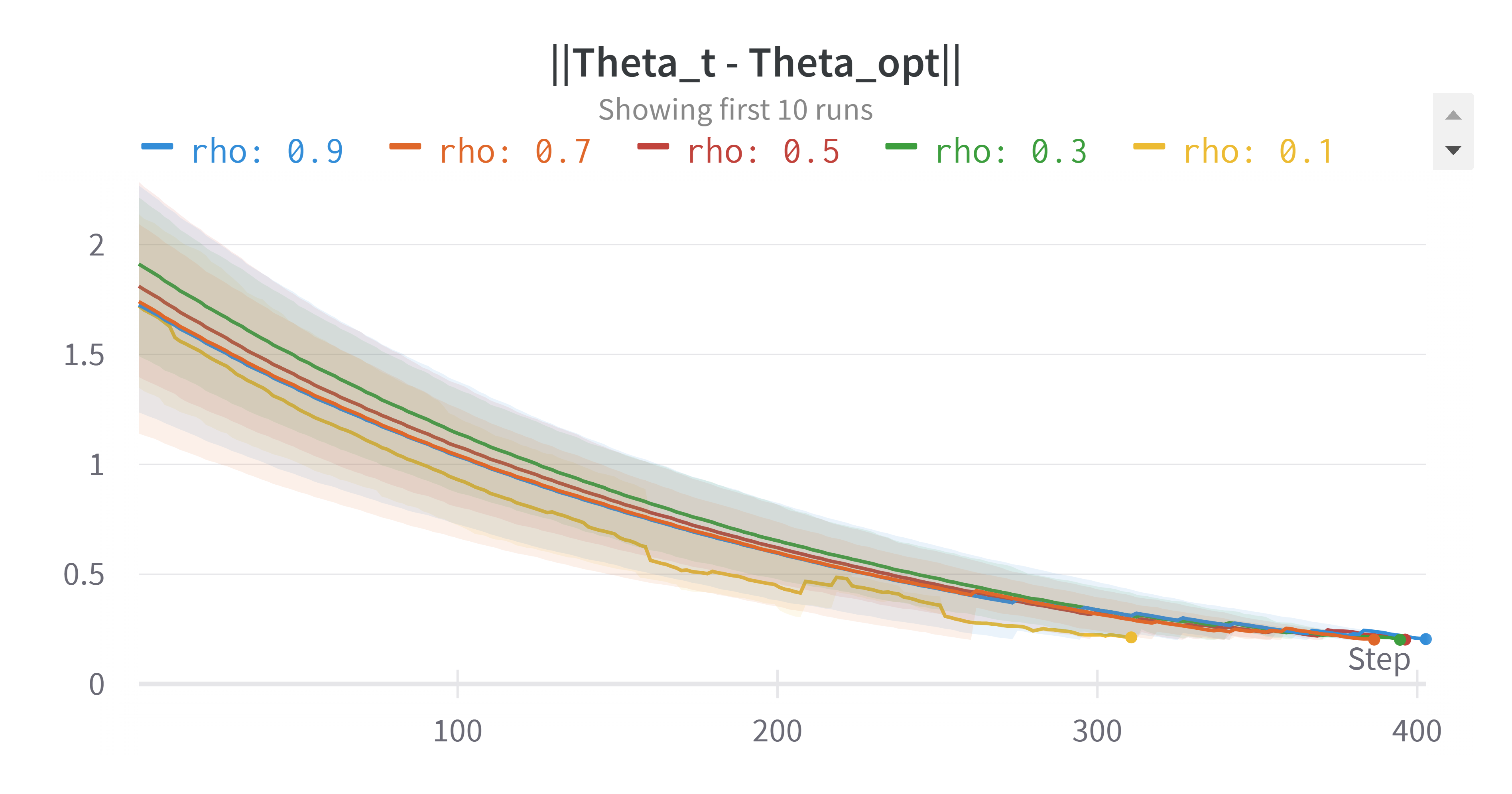}
		\caption{Spherical Gaussian}
	\end{subfigure}
	\hfill
	\begin{subfigure}[t]{0.475\textwidth}
		\centering
		\includegraphics[width=\textwidth]{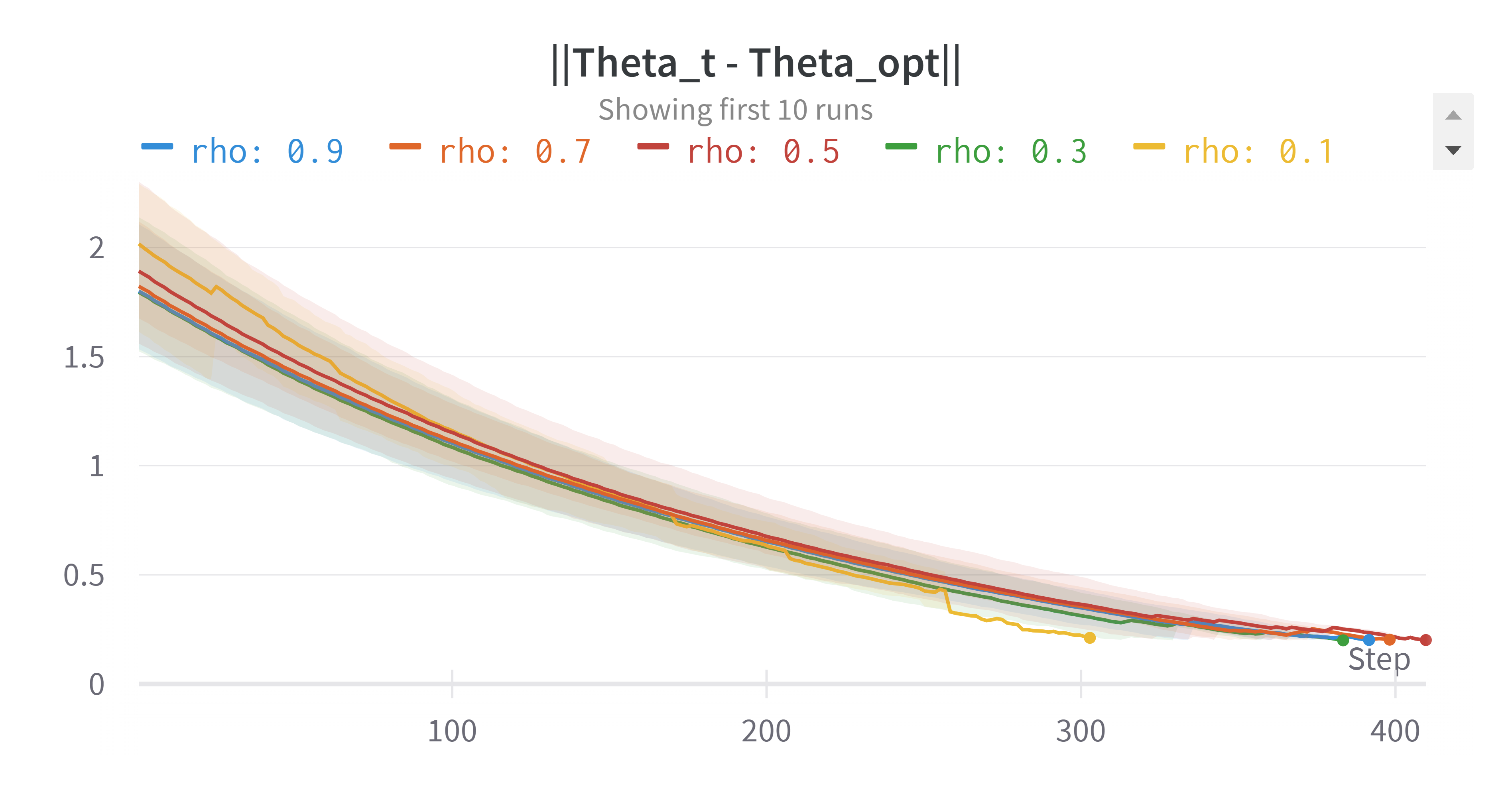}
		\caption{Conditional Gaussian}
	\end{subfigure}
	\vskip\baselineskip
	\begin{subfigure}[t]{0.475\textwidth}
		\centering
		\includegraphics[width=\textwidth]{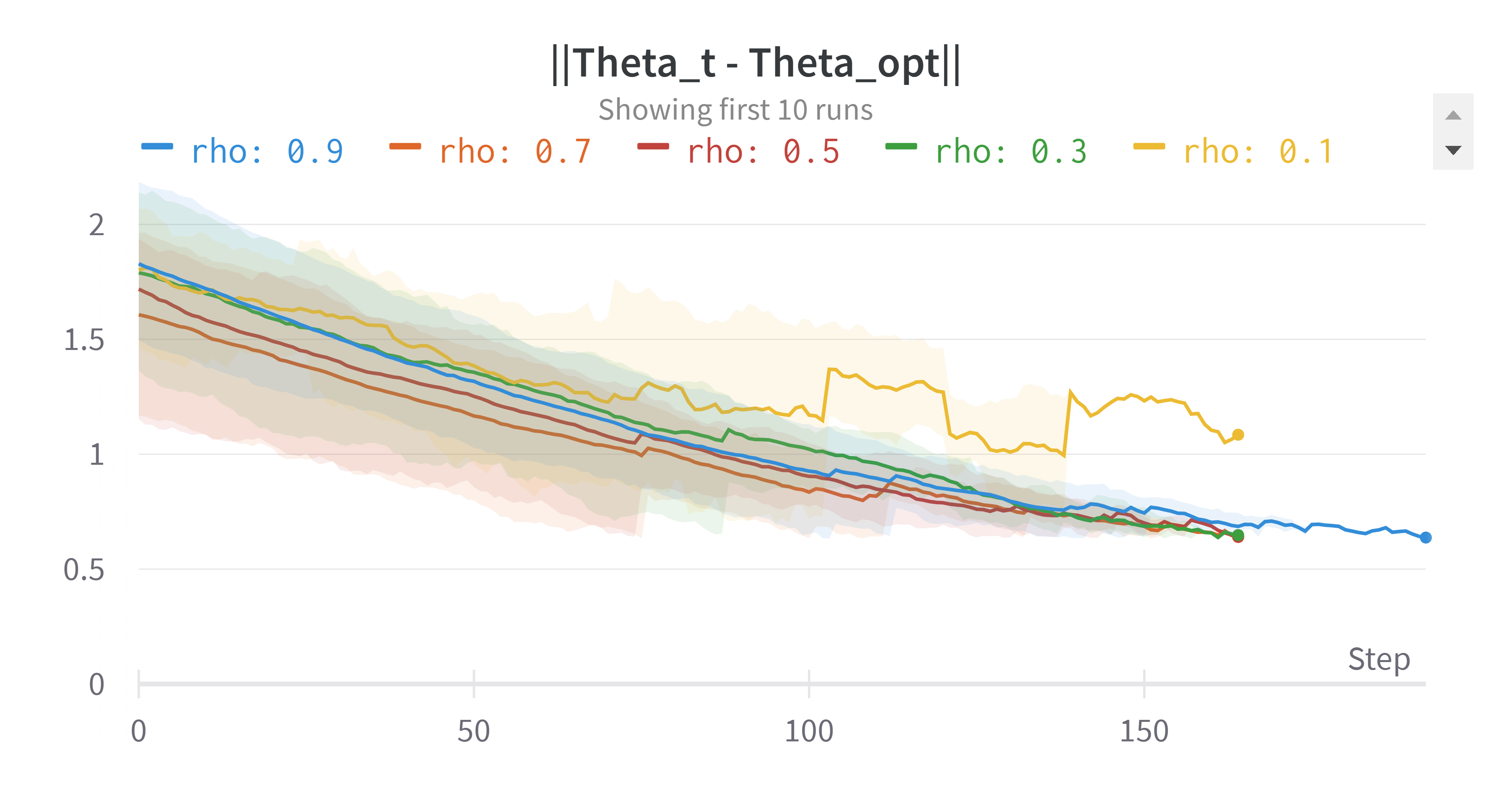}
		\caption{\label{subfig:productdist_rho}Product Distribution}
	\end{subfigure}
	\hfill
	\begin{subfigure}[t]{0.475\textwidth}
		\centering
		\includegraphics[width=\textwidth]{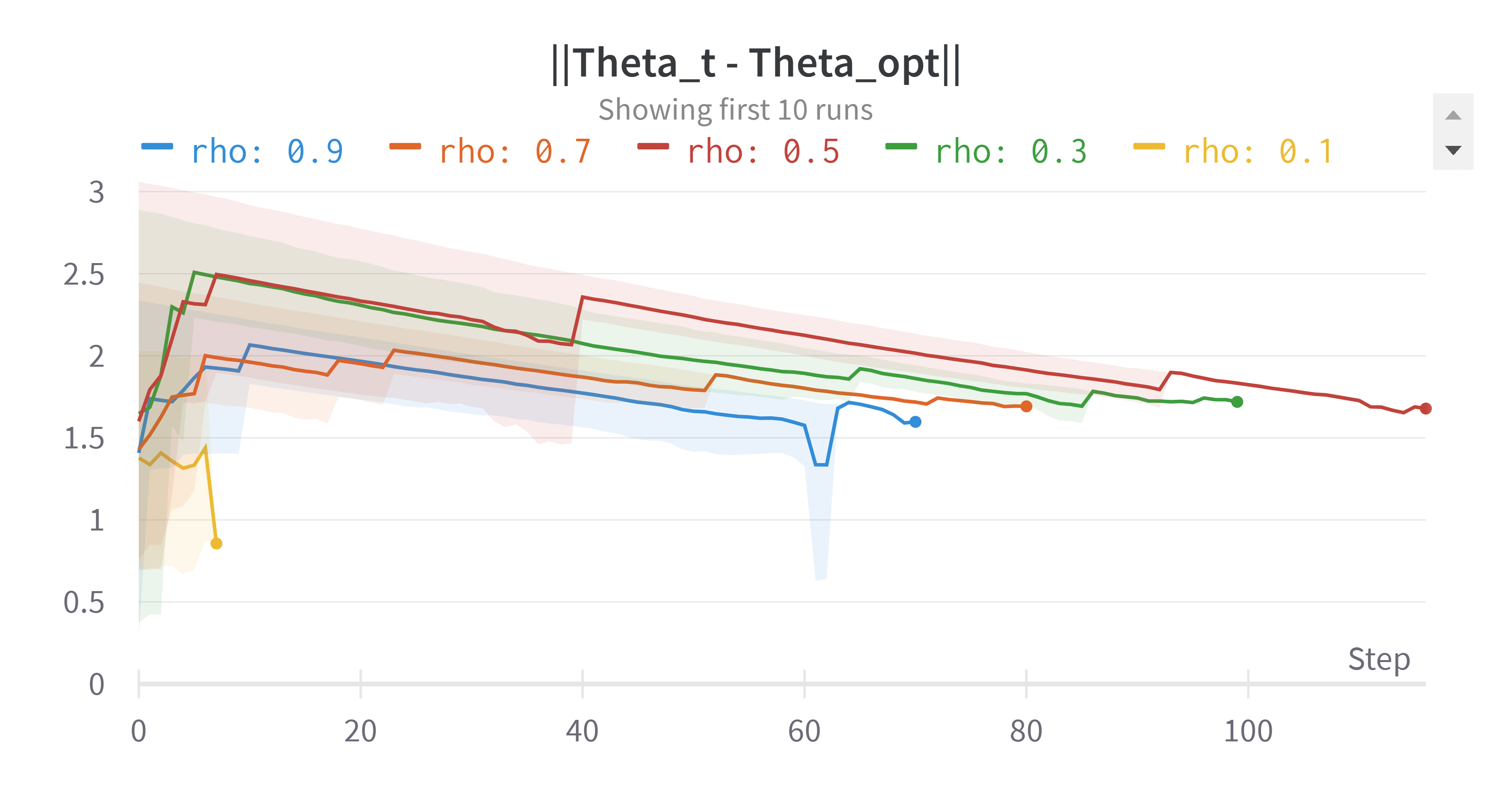}
		\caption{\label{subfig:barcrawl_rho}Bar Crawl: Detecting Heavy Drinking}
	\end{subfigure}
	\caption{\label{fig:results_by_rho} The distance of $\theta^t$ to $\theta_{opt}$ as a function of $t$ -- the iteration number, for $\gamma = 0.2$ and $\rho \in \{0.1, 0.3, 0.5, 0.7, 0.9 \}$. Each curve corresponds to a different $\rho$ value. In all experiments varying $\rho$ has a small effect on the halting time.} 
\end{figure}

\paragraph{Conclusions.}
Our experiments suggest that indeed our bound $T$ is a worst-case bound, where in all experiments we concluded in about $7-50$ times faster than the bound of Algorithm~\ref{alg:ngd-meb}.
This suggests that perhaps one would be better off if instead of partitioning the privacy budget equally across all $T$ iterations, they devise some sort of adaptive privacy budgeting. (E.g., using $\nicefrac {3\rho}{4}$ budget on the first $T/4$ iterations and then the remaining $\nicefrac \rho 4$ budget on the latter $\nicefrac {3T} 4$ iterations.) Such adaptive budgeting is simple when using zCDP, as it does not require ``privacy odometers''~\cite{RogersVRU16}.